\documentclass{article}
\usepackage[english]{babel}
\usepackage[round]{natbib}
\usepackage[section]{placeins} 
\usepackage[tracking=true]{microtype}  
\usepackage{amsmath} 
\usepackage{booktabs} 
\usepackage{siunitx} 
\usepackage{multirow} 
\usepackage{subfig}
\usepackage{threeparttable}
\usepackage{chngpage} 
\usepackage{afterpage} 
\newcolumntype{H}{>{\setbox0=\hbox\bgroup}c<{\egroup}@{}} 
\usepackage[tracking=true]{microtype}
\usepackage{booktabs}
\usepackage{subfig}
\usepackage{amsmath}%
\usepackage{amsfonts}%
\usepackage{amssymb}%
\usepackage{amsthm} 
\usepackage{graphicx}
\usepackage{algorithm}
\usepackage{algorithmic}
\usepackage{float}
\usepackage{color}
\usepackage{pdflscape}
\usepackage{enumitem}
\usepackage{caption} 
\usepackage[titletoc]{appendix}
\makeatletter
\renewcommand\paragraph{\@startsection{paragraph}{4}{\z@}%
            {-0.5ex\@plus -1ex \@minus -.25ex}%
            {0.25ex \@plus .25ex}%
            {\normalfont\normalsize}}
\makeatother
\makeatother
\usepackage{natbib}
\def\url#1{\expandafter\string\csname #1\endcsname}

\newtheorem{corollary}{Corollary}

\newtheorem{definition}{Definition}
\newtheorem{example}{Example}

\newtheorem{lemma}{Lemma}

\newtheorem{proposition}{Proposition}
\newtheorem{remark}{Remark}

\numberwithin{equation}{section}

\begin{document}
\title{The lateral transhipment  problem with a-priori routes, and a lot sizing application}
\author{\textbf{Martin~Romauch}\\Department of Business Administration\\University of Vienna, Austria\\martin.romauch@univie.ac.at\\ \\\textbf{Richard~F.~Hartl}\\Department of Business Administration\\University of Vienna, Austria\\richard.hartl@univie.ac.at\\ \\\textbf{Thibaut~Vidal}\\Departemento de Informática\\Pontifícia Universidade Católica do Rio de Janeiro, Brazil\\thibaut.vidal@mit.edu}
\date{Working Paper - Dec 2015}

\maketitle

\begin{abstract}
We propose exact solution approaches for a lateral transhipment problem which, given a pre-specified sequence of customers, seeks an optimal inventory redistribution plan considering travel costs and profits dependent on inventory levels. Trip-duration and vehicle-capacity constraints are also imposed.
The same problem arises in some lot sizing applications, in the presence of setup costs and equipment re-qualifications.

We introduce a pure dynamic programming approach and a branch-and-bound framework that combines dynamic programming with Lagrangian relaxation. Computational experiments are conducted to determine the most suitable solution approach for different instances, depending on their size, vehicle capacities and duration constraints.  
The branch-and-bound approach, in particular, solves problems with up to 50 delivery locations in less than ten seconds on a modern computer.
\end{abstract}

\maketitle

\section{Introduction and related work}

Optimization problems with combined inventory and routing decisions arise in a wide variety of contexts.
In inventory routing problems \cite{coelho2012dynamic}, for example, inventory and routing costs are minimized on a planning horizon. Each route occurs on a specific time period, originates from a central depot and visits some customers to replenish their inventories. The adequate selection of a subset of customers for each period, as in the team orienteering and prize-collecting problems \cite{chao96,Vidal2014}, is thus an essential problem feature.

Other related problems have been defined on a single planning period, such as the TSP with pickups and deliveries \cite{hernandez2004branch}, the balancing problems for static bike sharing systems \cite{vogel2011strategic,rainer2013balancing} and also the lateral transhipment problem for {a single route} (SRLTP), \cite{DBLP:conf/eurocast/HartlR13,SRLTP15}. This latter problem aims to distribute inventory on a network $(V,\mathcal{A})$ via pickups and deliveries using {one} vehicle, to minimize a non-linear objective.
In bike sharing systems \cite{raineretal13}, a target level is defined and the objective is to minimize the corresponding total deviation (a piecewise-linear convex function). Some MIP formulations of this problem are introduced in \cite{raviv2013static}. The objective includes expected shortage costs and travel costs, similar to the TSP with profits~\cite{feillet05}. Dynamic route interactions like hand-overs (intermediate storage) and multiple visits are also considered.

In the context of the SRLTP, both travel costs and profits must be considered. Each location $i \in V$ is characterized by a piecewise-linear \emph{profit function} $F_i$. The problem is to find inventory changes $y_i$, such that the revenue minus the costs for the pickup-and-delivery routes is maximized. Suppose that $I_0^{i}$ is the initial inventory and $I_{min}^{i}$, $I_{max}^{i}$ are bounds on the inventory level, then the total revenue can be expressed as:
\begin{equation}
\sum_{i \in I} F_i(I_0^{i} - y_i) \quad \text{for }  I_0^{i} - y_i \in [I_{min}^{i},I_{max}^{i}]. \label{formula:F}
\end{equation}
Each vehicle, subject to a load limit $Q_{max}$ and a tour length limit $T_{max}$, visits a subset $S \subset V$ of customers to redistribute their inventories. The travel cost and duration on an arc $(i,j)$ is notated as $c_{ij}$ and $t_{ij}$.

This work aims to contribute towards better addressing the SRLTP, through a dedicated study of one essential subproblem: the a-priori route evaluation problem for lateral transhipment (ARELTP). Indeed, most modern heuristic techniques for the SRLTP consider a large set of vehicle routes during the search and aim to evaluate their profit. The goal of this paper is to find an efficient algorithm, which, for a given route (i.e., a sequence of visits), returns the optimal pickup or delivery quantities at each location in the presence of piecewise-linear profit functions, capacity and distance constraints. We also consider the ability to \emph{shortcut} a customer if this is profitable.

The problem is also very relevant on its own, as a case of routing optimization with a-priori routes \cite{bertsimas1993further}. In practical routing applications, retaining some fixed route fragments
 can lead to a better operational and computational tractability for companies, as well as efficiency gains through driver learning.
The corresponding subproblem is called the evaluation problem for a-priori routes, and efficient solution methods are needed to quickly react to changing environments. We also show that the same model encompasses several lot sizing applications with re-qualification costs.

The contributions of this paper are the following.
We first provide a formal definition of the lateral transhipment problem with a-priori routes.
In contrast with previous articles on the topic, general piecewise linear cost functions are considered. This enables to model economies of scale and expected stochastic demands.
To address this problem, we introduce a pure dynamic programming approach and a branch-and-bound framework that combines dynamic programming with Lagrangian relaxation. The methods are also adapted for problem settings where the triangle inequality is not satisfied, hence allowing to deal with lot sizing models where the triangle inequality (for setup costs and times) is often violated.
Finally, extensive computational experiments are conducted to determine the most suitable solution approach for different instances, depending on their size and the magnitude of some of their key parameters, e.g., vehicle capacities. The resulting algorithms finds optimal solutions for small- and medium-scale instances in a fraction of seconds.

The paper is organized as follows, in Section~\ref{sec:probdesc}, the ARELTP is formally defined and its computational complexity is analyzed. The lot sizing application with re-qualification costs (LSwRC) is also presented in Section~\ref{sec:LSwRC} and alternative mixed integer linear programming models for the problems are discussed in Section~\ref{sec:MIPsolverformualtion}. 
The proposed dynamic programming and branch-and-bound approaches are described in Sections~\ref{sec:DP} and \ref{sec:BBLR}. To impact of the absence of triangle inequalities is discussed in Section~\ref{sec::triangle}. Finally,  Section~\ref{sec:comp} describes our computational experiments and Section~\ref{sec:conclu} concludes.

\section{Problem statement} 
\label{sec:probdesc}

This section introduces a mathematical formulation of the ARELTP and discusses its computational complexity. Let $\tau = (1,\ldots,n)$ be an a-priori route, i.e., a sequence of $n$ locations. 
Now, suppose that some of the locations in $\tau$ are skipped in the optimal subtour $\tau^*=(\tau_{i_1},\tau_{i_2},\ldots,\tau_{i_m})$ and the optimal inventory changes are $y_i^*$, then the total profit minus routing costs is equal to:
$$ \sum_{i=1}^n F_i(I_0^{i} + y_i^*) - \sum_{l=1}^{m-1} c_{\tau_{i_l},\tau_{i_{l+1}}}.$$ 
If the total optimal revenue is larger than $\sum_{i=1}^n F_i(I_0^{i})$, then the lateral transhipment on route $\tau$ route is profitable, otherwise not. 

To simplify the notation, we define the cost change $f_i$ at a location $i$ as
$f_i(y_i) = F_i(I_0^{i}) - F_i(I_0^{i} - y_i)$, as a function of the inventory change $y_i$, which should be in the interval $[a_i,b_i]$, where $a_i = I_0^{i}-I_{max}^{i}$ and $b_i = I_0^{i}-I_{min}^{i}$. Note that $a_i \leq 0 \leq b_i$.
The ARELTP can then be formulated as:
\begin{align}
\min \hspace*{0.2cm} \psi(x,y) = &\sum_{1 \leq i < j \leq n}  c_{ij} x_{ij} + \sum_{1 \leq i \leq n} f_i(y_i)   \label{ARELTP:obj} \\
s.t.\quad & \sum_{j<i} x_{ji} =  \sum_{j>i} x_{ij} &\quad 1 < i < n  \label{ARELTP:FB}\\
& \sum_{j>1} x_{1j} = 1  \label{ARELTP:source}\\
& \sum_{j<n} x_{jn} = 1  \label{ARELTP:sink}\\ 
& a_i \sum_{j>i} x_{ij}  \leq y_i \leq b_i \sum_{j>i} x_{ij}  & \quad 1 \leq i \leq n  \label{ARELTP:visit} \\
& 0  \leq  \sum_{j \leq i} y_j \leq Q_{max} & \quad 1 \leq i \leq n \label{ARELTP:Qmax}  \\
& x_{ij} \in \{0,1\} & \quad 1 \leq i<j \leq n \label{ARELTP:xbin} \\ 
& y_i \in \mathbb{R} & \quad 1 \leq i \leq n   \label{ARELTP:y}
\end{align}

The objective (\ref{ARELTP:obj}) is equivalent to maximizing the total profit minus the distance cost. 
The arc selection variables $x_{ij}$ are defined for $i<j$, therefore it is sufficient to formulate the flow balance (\ref{ARELTP:FB}) and the constraints for the source (\ref{ARELTP:source}) and the sink (\ref{ARELTP:sink}) to define a subsequence of $\tau$. 
According to (\ref{ARELTP:visit}), changing the inventory level at a location $i$ ($y_i \neq 0$) implies that the location must also be visited.
The load of the truck when leaving $i$ is $\sum_{j \leq i } y_j$, therefore (\ref{ARELTP:Qmax}) enforces that $Q_{max}$ is the corresponding upper limit. The MIP (\ref{ARELTP:obj}-\ref{ARELTP:y}) is called the ARELTP without duration limit constraint.

We will also consider a variant of the problem (\ref{ARELTP:obj}--\ref{ARELTP:y}), called the ARELTP with duration limit, with the following additional constraint:
\begin{align}
& \sum_{1 \leq i < j \leq n}  t_{ij} x_{ij} \leq T_{max}.  \label{ARELTP:Tmax} 
\end{align}

In previous literature, a simplified version of the ARELTP was presented in \mbox{\cite[][Section~3.3]{DBLP:conf/eurocast/HartlR13}}, but subject to three simplifications:
\begin{itemize}[nosep]
 \item $c_{ij}$ is not considered,
 \item $f_i$ is linear,
 \item $T_{max}$ is not considered.
\end{itemize}

By transforming the problem to a minimum cost flow, it is possible to solve this special case in polynomial time. However, revoking any one these simplifications leads to a NP-hard problem. As demonstrated in Section \ref{sec:complexity} of the appendix, the ARELTP is NP-hard if either $c_{ij}$ is considered, or $f_i$ is piecewise linear, or if $T_{max}$ is considered.

In the following, two exact solution approaches for the ARELTP (\ref{ARELTP:obj}-\ref{ARELTP:Tmax}) will be proposed. The special case without duration constraint (\ref{ARELTP:obj}-\ref{ARELTP:y}) will also be discussed separately.
The ARELTP is an interesting an difficult problem on its own. In particular, a lot sizing problem with re-qualification costs (LSwRC) is described in the next section, and a transformation that establishes the equivalence of the problems is provided.

\section{Lot sizing with requalification costs} 
\label{sec:LSwRC}

Various algorithms are known to deal with lot sizing models considering product-dependent setup times \cite{drexl1997lot, salomon1997solving}. In contrast, lot sizing problems with idle time-dependent setups have been much less studied. Such models arise in food and pharmaceutical industry, where the qualification of processes and tools have a given duration or expiration (see \cite[][5.4.j]{Bedson96} and \cite{bennett2003pharmaceutical}). 
In these applications, a frequent use of a tool may stretch the duration of its qualification. 
This characteristic also appears in semiconductor industry, where tools need to be qualified for each process and product. First-time qualifications are usually time consuming and expensive. For some tools (e.g. steppers in lithography), this qualification for a product expires if it is not running on this tool for a longer period and expensive requalifications become necessary. 

We will formulate a Lot Sizing problem with Re-qualification Costs (LSwRC) which considers these aspects, and the equivalence with the ARELTP will be established. 
Consider the decision variable $y_i$ for the production quantity in period $i$. If production takes place in period $i$ ($y_i>0$) the lower bound $a_i$ and the upper bound $b_i$ need to be respected.

The production cost for each period $i$ is represented by a piecewise linear functions $f_i(y_i)$, therefore the total production cost is $\sum_i f_i(y_i)$. The inventory level $q_i$ at the end of period~$i$ and the unit holding cost $h_i$ for holding one item in period $i$ for one period defines the inventory holding cost $\sum_i h_i q_i$. The inventory level is zero in the beginning ($q_0 = 0$) and the balance equation $q_{i} = q_{i-1} + (y_i- d_i)$ states that the inventory level $q_i$ at the end of period $i$ is non-negative. In other words, the demand is satisfied at all times, i.e. $q_{i} = \sum_{j \leq i} (y_{j} - d_{j}) \geq 0$. We also introduce setup variables to model time-dependent setup costs: $x_{ij} \in \{0,1\}$ being valued to one if and only if $y_i > 0$ and $y_j > 0$. The corresponding setup cost $c_{ij}$ covers the qualification costs and may be larger for long idle times, when considering expensive re-qualifications and setups.  

Finally, associating a resource consumption $t_{ij}$ if $x_{ij}=1$, and setting a restriction on the total setup-related expenses $T_{max}$ leads to the same constraint as Equation (\ref{ARELTP:Tmax}). This restriction can be used to limit the number of production periods, by setting $t_{ij}=1$ for all $(i,j)$, or to limit the total setup cost by using $t_{ij} = c_{ij}$. The complete model can be stated~as:
\begin{align}
 \min & \sum_{1 \leq i < j \leq n}  c_{ij}x_{ij}  + \sum_{1 \leq i \leq n} f_i(y_i) + \sum_{1 \leq i \leq n} h_i q_i  \label{LSwRC:obj} \\
s.t.\quad & q_i = \sum_{j \leq i} (y_j-d_j) & \quad 1 \leq i \leq n \label{LSwRC:inventorylevel} \\ 
& a_i \sum_{j>i} x_{ij}  \leq y_i  \leq b_i \sum_{j>i} x_{ij}  & \quad 1 \leq i \leq n  \label{LSwRC:feasProd} \\
& 0  \leq  q_i \leq Q_{max} & \quad 1 \leq i \leq n \label{LSwRC:Qmax}  \\
& x_{ij} \in \{0,1\} & \quad 1 \leq i<j \leq n \label{LSwRC:bin} \\
& y_i \in \mathbb{R} & \quad 1 \leq i \leq n \label{LSwRC:cont} \\
& \text{including} \quad (\ref{ARELTP:FB}-\ref{ARELTP:sink})\text{, }(\ref{ARELTP:Tmax})  \nonumber 
\end{align}

To evaluate inventory costs as a function of $y_i$, the following conversion~is~used:
\begin{align}
\sum_{i=1}^n h_i q_i = \sum_{i=1}^{n} H_{i,n} (y_i-d_i) \quad \text{where} \quad H_{i,n} = \sum_{j=i}^n h_j.  \label{LSwRC2:obj}
\end{align}

The following substitutions enable then to reduce the LSwRC to the ARELTP:
\begin{itemize}[nosep]
\item $y_i \gets y_i' + d_i$, and 
\item $f_i(y_i) + H_{i,n} (y_i-d_i) \gets f_i'(y_i')$.
\end{itemize}

Too facilitate the exposition, we will present the different solution approaches based on the ARELTP terminology. Our computational experiments will also cover some instances of the LSwRC.
In the following, we will discuss different exact solution approaches. Section~\ref{sec:MIPsolverformualtion} first provides a quick discussion about the direct resolution of this model with standard MIP solvers. Alternative linearizations are discussed, since they have a significant impact on performance.

\section{MIP formulation for piecewise linear costs} 
\label{sec:MIPsolverformualtion}

In order to solve the problem with a state of the art MIP solver, a suitable formulation for the problem (\ref{ARELTP:obj}-\ref{ARELTP:Tmax}) is provided in this section. The formulation is based on special ordered sets (SOS2) which an established approach for linear problems (\cite{vielma2008nonconvex},\cite{keha2006branch}) and mixed integer problems \cite{moccia2011modeling}.

If $f_i$ is piecewise linear in the general sense, it may happen that \mbox{$\liminf_{y \to a} f_i(y) < f_i(a)$}
for points $a$ on the border of neighboring interval domains, and strict inequalities ($\epsilon$-constraints) are necessary to apply linear programming techniques. Therefore, we assume that the profit change functions $f_i$ are piecewise linear and lower semicontinuous, which allows to extend the domains of all segments to the closed intervals without causing ambiguities.

According to that, each segment is represented by the convex combination of the points related to the interval borders, and $f_i$ can be represented as a sequence of points $\{(X_{i,0},Y_{i,0}),\ldots,(X_{i,m_i},Y_{i,m_i})\}$.

In the following formulation the decision variables $X_i$ and $Y_i$ will used to model $f_i$, i.e. $f_i(X_i) = Y_i$. According to the choice of $X_i$ one of the segments will be activated (type $SOS2$), by choosing neighboring points and a corresponding convex combinations $X_i = \lambda_{i,k} X_{i,k}+ \lambda_{i,k+1} X_{i,k+1}$, i.e.:
\begin{align}
& \sum_{k=0}^{m_i} \lambda_{i,k} X_{i,k} = X_i & \quad 1 \leq i \leq n \\
& \sum_{k=0}^{m_i} \lambda_{i,k} Y_{i,k} = Y_i & \quad 1 \leq i \leq n 
\end{align}  
Accordingly, the objective has the following form: 
$$ \sum_{1 \leq i < j \leq n}  c_{ij} x_{ij} - \sum_{1 \leq i \leq n} Y_i $$

If SOS2 is not included in the modeling language, two alternatives are presented, 
the first one uses a set of variables $\alpha_{i,k}$ to identify the selected neighboring pairs of points and the second one solves this problem by adding constraints on non-subsequent pairs of points. The first alternative is formulated as follows:
\begin{align}
& \sum_{k=0}^{m_i} \lambda_{i,k} = 1 & \quad 1 \leq i \leq n \label{SOS:1}\\
& \lambda_{i,k} \leq \alpha_{i,k}  & \quad \quad  1 \leq i \leq n, \ 0 \leq k \leq m_i  \label{SOS:2} \\
& \sum_{k=0}^{m_i} \alpha_{i,k}  = 2 &\quad \quad  1 \leq i \leq n  \label{SOS:3}\\
& \alpha_{i,k} \in \{0,1\}   &\quad \quad  1 \leq i \leq n , \ 0 \leq k \leq m_i  \label{SOS:5}\\
& \alpha_{i,k} + \alpha_{i,k'} \leq 1 & \quad \quad  1 \leq i \leq n, \ 0 \leq k \leq m_i-2 , \ k+2 \leq k' \leq m_i \label{SOS:6} \\
& \lambda_{i,k} \in [0,1] &\quad \quad  1 \leq i \leq n , \ 0 \leq k \leq m_i  \label{SOS:7} 
\end{align}
In the second alternative, the binary variables $\beta_{i,k}$ are introduced to replace (\ref{SOS:6}) by \mbox{(\ref{SOS:8}--\ref{SOS:9})}:
\begin{align}
& \sum_{k=0}^{m_i-1} \beta_{i,k}  = 1 &\quad \quad  1 \leq i \leq n  \label{SOS:4}\\
& 2 \beta_{i,k} \leq \alpha_{i,k} + \alpha_{i,k+1} &\quad \quad  1 \leq i \leq n , \ 0 \leq k \leq m_i-1  \label{SOS:8}\\
& \beta_{i,k} \in \{0,1\}   &\quad \quad  1 \leq i \leq n , \ 0 \leq k \leq m_i-1  \label{SOS:9}
\end{align}
According to our computational experiments, the second alternative leads to a smaller CPU time than the first. The differences between the second alternative and using SOS2 directly are not significant. We thus considered the second alternative in our computational experiments of Section \ref{sec:comp}. 
The next sections introduce two new solution approaches: a pure dynamic programming algorithm and a branch-and-bound approach based on a Lagrangian relaxation of the route duration constraint.

\section{Dynamic programming approach}
 \label{sec:DP}

This section introduces two dynamic programming approaches for the ARELTP. We assume for now that the triangle inequalities for $c_{ij}$ and $t_{ij}$ are satisfied and that no duration constraints are imposed.

\subsection{DP -- Without duration constraint} 
\label{sec:DPnoTmax}

We define the value function $V_k(q)$ for $q\geq 0$, which returns the optimal cost of a route ending at $k$ with inventory level $q \geq 0$. The functions $V_k$ can be expressed as:
\begin{align}
V_k(q) =\min & \sum_{i < j \leq k}  c_{ij} x_{ij} + \sum_{i \leq k} f_i(y_i)   \label{PAD3:obj} \\
& q_k = \sum_{i \leq k} y_{i} = q \\
& \text{(\ref{ARELTP:FB}-\ref{ARELTP:y}) when replacing } n \text{ by } k.  \label{PAD3:st}
\end{align}
These functions can be computed by using the following recurrence relation:
\begin{align}
V_{1}(q) = f_1(q), \quad q \geq 0  \label{fx:init}\\
V_{i} = \min_{j < i} \left( \overbrace{\min_{\underset{q-y \in \mathcal{D}(V)}{y \in \mathcal{D}(f_i)}}\{ V_j(q-y) + c_{j,i} + f_i(y) \}}^{V_{ji}:=} \right). \label{fx:rec}
\end{align}

In the case where skipping one or more visits is forbidden, then $V_{i-1}$ would be the only predecessor of $V_i$. We now define the \emph{superposition} operation $\boxplus$ as
\begin{equation} (V \boxplus f) (q) = \min_{\underset{q-y \in \mathcal{D}(V)}{y \in \mathcal{D}(f)}}\{ V(q-y) +  f(y) \}, \label{eq:superos} 
\end{equation}
and the recurrence relation (\ref{fx:rec}) gets the following form:
\begin{align}
V_{i} = \min_{j < i} ( (V_j + c_{j,i}) \boxplus  f_{i} ). \label{f:rec}
\end{align}

The value functions are piecewise linear as a consequence of the MIP representation \mbox{(\ref{PAD3:obj}--\ref{PAD3:st})}. We thus represent each piecewise linear function $f$ as a sequence of segments $v_l^{f}$ for $l = 1, \ldots, l^{f}_{max}$.
Our algorithm computes iteratively the results of the dynamic programming recurrence (\ref{f:rec}), and stores the values functions as lists of functions segments.
Each iteration is performed in two steps. 
First, $\tilde{V}_i = \min_{j < i} (V_j + c_{j,i})$ is computed (called the \emph{envelope} of the functions $V_j + c_{j,i}$) via a simple procedure which compares segments with a common domain.
Second, the \emph{superposition} $\tilde{V}_i \boxplus f_{i}$ is evaluated. 

\begin{figure}[!htb]
\centering
\includegraphics[width= 0.7\hsize]{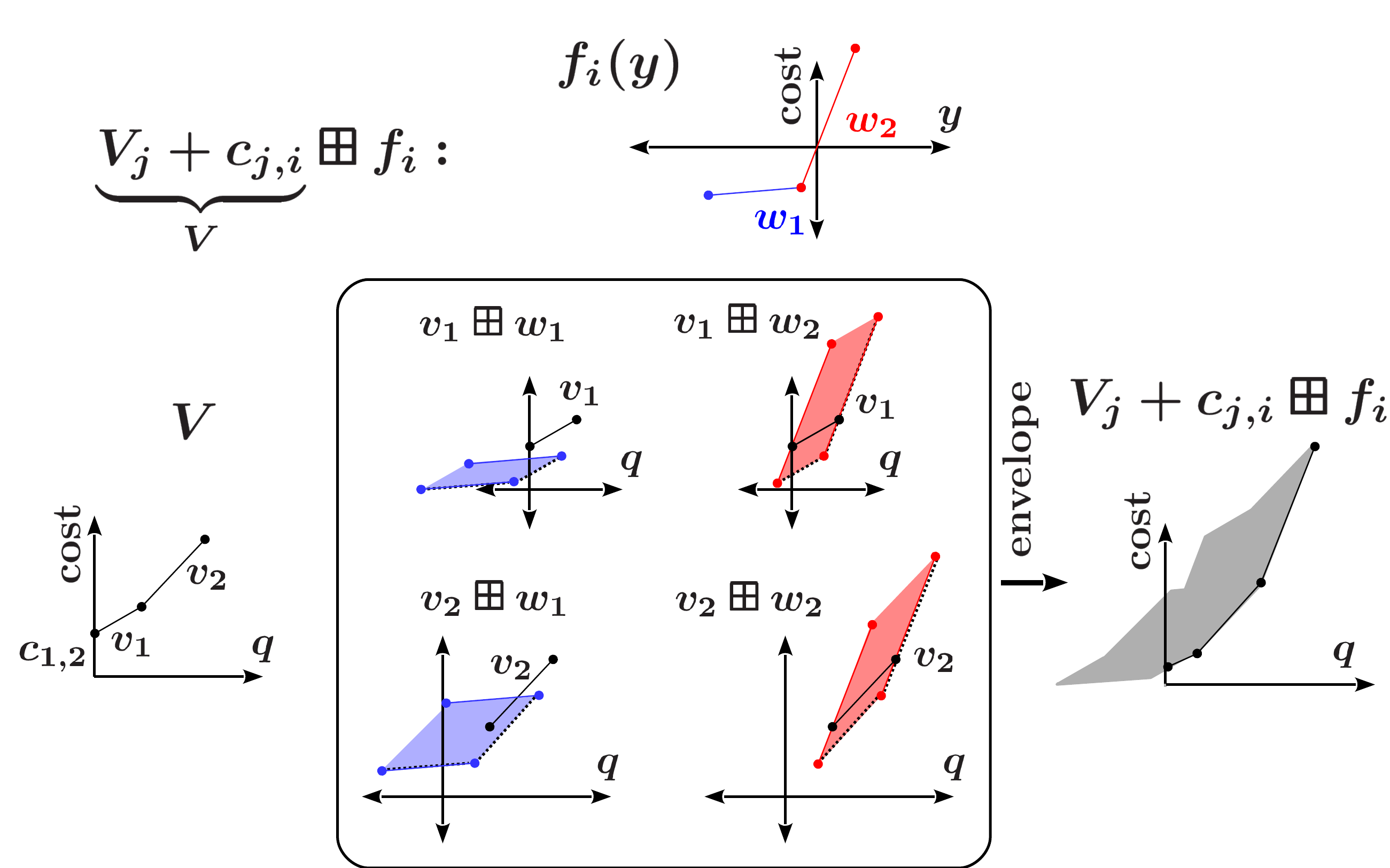}
\caption{An example for calculating the superposition of piecewise linear functions}
\label{fig:DPrec2}
\end{figure}

Figure~\ref{fig:DPrec2} provides an example of the evaluation of the superposition \mbox{$\tilde{V}_i \boxplus f_{i}$}, using functions with two segments each. This is done by evaluating the superposition of all pairs of segments $(v_l^V,v_{l'}^f)$ from $V$ and $f$, respectively, and extracting the minimum (lower envelope) of the resulting piecewise-linear functions. In the case of Figure~\ref{fig:DPrec2}, four superpositions of segments are performed.
The result of the superposition of two segments $v_l^V \boxplus v_{l'}^f (q)= \min_y \{ v_l^V(q-y) + v_{l'}^f(y) \}$ can be derived analytically in $\mathcal{O}(1)$. A geometrical interpretation of this step is that the set of feasible $(q,y)$ values forms a parallelogram, and the result of a superposition corresponds to the lower edges of these parallelograms. 

Overall, the complexity of one iteration of the dynamic programming -- computing $\tilde{V}_i$ from the value functions of predecessors -- is linked to the complexity of the \emph{envelope} operation and the complexity of the \emph{superposition} operation.
Let $m_i$ be the number of segments of the value function $V_i$ and $M_i=\sum_{j < i} m_j$ the total number of segments of the value functions of predecessors $i < j$. Let $\hat{m}_i$ be the number of segments in $f$. As shown in Section \ref{appendix:superpos} of the appendix, the superposition can be implemented in $\mathcal{O}(\log( \hat{m}_i) M_i \hat{m}_i )$ and according to Section \ref{appendix:envelope}, the envelope operation can be computed in $\mathcal{O}(\log(i) M_i \alpha(M_i))$, where $\alpha$ denotes the inverse Ackermann function. The superposition operation is a special case of the calculation of the envelope of the Minkowski sum of two polygons (cf. \cite{agarwal2002polygon},\cite{ramkumar1996algorithm}).

Finally, note that in many applications, the capacity $Q_{max}$ and the domains of the segments of $f_i$ are integer. This allows to restrict the segments of the value function to integer domains  reduce the number of segments of the values functions, as explained in Section \ref{sec:integerdata} of the appendix.  

\subsection{DP -- Considering the duration constraint}
\label{sec:DPTmax}
This section now describes the dynamic programming recursion when considering the duration constraint $T_{max}$. The value functions now include one additional dimension related to the duration $t$.
Let $V_{i,t}(q)$ be the value function for step $i$ and for travel times equal to $t$, which returns the best profit for a final inventory $q$. The recurrence relation for $V_{i,t}$ has the following form:  
\begin{align}
V_{1,0}(q) &=   f_i(q) \quad & q \geq 0 \\
V_{i,t} &=   \min_{j < i} ( V_{j,t-t_{j,i}} + c_{j,i}) \boxplus f_{i} \label{formulat:DPt0}
\end{align}

With these conventions, $V_{i,t}$ is defined for every possible duration considering all feasible paths of different lengths. This number of possible durations values grows exponentially and can include dominated elements. Hence,
to avoid representing dominated parts, we compute an alternative value function $U_{i,T}$, where only the non-dominated options $V_{i,t}(q)$ for a given \emph{duration budget} $T$ are considered. This function is defined as follows:
$$U_{i,T}(q) :=   \min_{t \leq T} ( V_{i,t}(q) ).$$
Note that $U_{i,t}(q)$ is piecewise constant in $t$ if $q$ is fixed. Furthermore, if the profit functions are monotonically decreasing, then the functions $U_{i,t}$ are monotonically increasing in~$q$, and monotonically decreasing in $t$.
The functions $U_{i,T}$ can be evaluated via the following recurrence formula:
\begin{align}
& U_{1,0}(q) =   f_i(q) \quad q \geq 0  \label{formulat:DPt1}\\
& U_{i,t} = \min_{j<i, t'+ t_{ji} \leq t} (U_{j,t'} + c_{j,i}) \boxplus f_{i}. \label{formulat:DPt2}
\end{align}

The values functions $U_{i,t}$ are computed explicitly in our dynamic programming algorithm. As in the previous section, each step of the recursion requires to apply the \emph{envelope} and \emph{superposition} operations. These operations are repeated for distinct values of $t$.
The explicit representation of the value functions is also generalized : instead of linked lists of function segments, the algorithm relies on several connected lists, one for each non-dominated duration value $t$. 

This structure is illustrated in Figure \ref{fig:MLLa}. The corresponding value function $U_{i,t}$ consists of planes that are constant with respect to $t$. Each plane is represented by a single segment for the corresponding domain of feasible time budgets. The segments and the corresponding domains have the same color; darker colors correspond to larger time budgets.
Figure \ref{fig:MLLb} is a representation of $U_{i,t}$ by linked segments. Each class of time budgets (e.g. $t \in [t_3,t_4)$) is associated with a unique color and the corresponding linked list of valid segments can be traced by following the segments and links of this color. There are two types of links, down-links that lead to segments that are feasible for a lower time budgets, and up-links that lead to segments that are feasible for higher time budgets. This structure can be used to traverse $U_{i,t}$ with respect to $q$ for a fixed value $t$.

  \begin{figure}[!ht]
    \subfloat[Value function $U_{i,t}$  \label{fig:MLLa}]{%
			\includegraphics[width=0.45\textwidth,page=1]{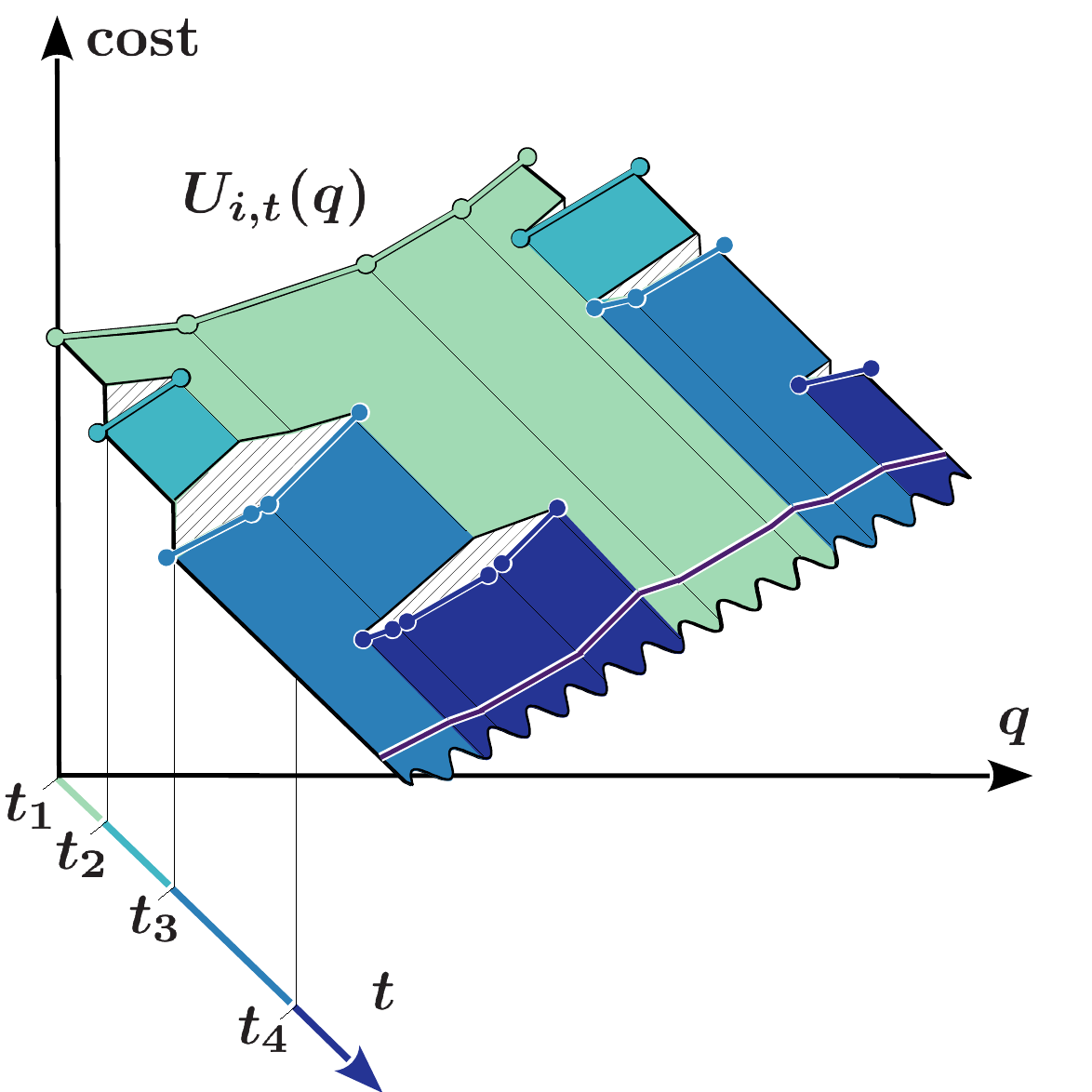}		
    }
		\hspace*{0.4em}
    \subfloat[Representation using nested lists \label{fig:MLLb}]{
			\includegraphics[width=0.45\textwidth,page=2]{MLLbw.pdf}
    }
    \caption{Example: value function with respect to segments and duration budgets $t$}
    \label{fig:dummy}
  \end{figure}

\section{Branch and bound with Lagrangian relaxation} 
\label{sec:BBLR}

Adding the duration constraint in the dynamic programming algorithm results in a larger number of non-dominated function segments. A key difference between the duration and capacity constraints is that the route duration limit $T_{max}$ acts as a global constraint, and there are little opportunities to prune by feasibility in intermediate steps. To circumvent this issue, we introduce an alternative algorithm based on a Lagrangian relaxation \cite{fisher2004lagrangian} of the duration constraint.

The notation will be extended to simplify the description of the methods.
Let $\psi(x,y)$ be the objective value (\ref{ARELTP:obj}) for a feasible solution $(x,y)$, and $\psi(x) = \min_y \psi(x,y)$ (indicating that finding the optimal $y$ is declared as the subordinate problem). Note that finding the corresponding optimal solution for $y$ can still be a hard problem (Section \ref{sec:complexity}). 
Furthermore, each sub-route $\sigma$ of $\tau$ can be characterized by its subset of visited customers, since the order of indices also defines the order of customers. The subset notation will thus also be used for subsequences. The objective value for a subroute $\sigma$ is identified as $\psi(\sigma)$.
\subsection{Lagrangian relaxation}
\label{sec:LR}

Applying a Lagrangian relaxation to the duration constraints leads to the following formulation,
\begin{align}
\max_{\lambda \geq 0} \left(
\begin{array}{l} \displaystyle \underset{x,y}{\min}  \sum_{1 \leq i < j \leq n}  c_{ij} x_{ij} + \sum_{1 \leq i \leq n} f_i(y_i) - \lambda \overbrace{(T_{max} - \sum_{1 \leq i < j \leq n}  t_{ij} x_{ij})}^{s(x):=} \\
\text{s.t.  (\ref{ARELTP:FB}--\ref{ARELTP:y})}
\end{array} \right), \label{f:LR1}
\end{align}
which is equivalent to minimizing $L(\lambda)$, for $\lambda \geq 0$, such that 
\begin{align}
 L(\lambda) &= -\lambda T_{max} +  \left( \begin{array}{l} \displaystyle \underset{x,y}{\min} \sum_{1 \leq i < j \leq n}  \left(c_{ij} + \lambda t_{ij} \right) x_{ij} + \sum_{1 \leq i \leq n} f_i(y_i)   \\
\text{s.t.  (\ref{ARELTP:FB}--\ref{ARELTP:y})}  \end{array}\right) \label{f:LR3} 
\\
x_\lambda:&\quad \psi(x_\lambda) - \lambda s(x_\lambda) =  L(\lambda). \label{f:LR3arg} 
\end{align}

We note that $(c_{ij} + \lambda t_{ij})$ is again a distance matrix and the triangle inequality holds.
The dynamic programming algorithm of Section \ref{sec:DPnoTmax}
 can thus be used to evaluate $L(\lambda)$ for any given value of $\lambda$, and the corresponding primal solution is called $x_\lambda$.
The function $L(\lambda)$ is concave \cite[see, e.g.,][]{fisher2004lagrangian}, and any $L(\lambda)$ is a lower bound of (\ref{ARELTP:obj}-\ref{ARELTP:y}). 

The best dual bound $L(\lambda^*)$ can be found by means of standard iterative techniques for convex optimization, such as bisection search, intersection of tangents, or quasi-Newton, among others. 
These methods depend on an initial search interval $I_\lambda = [\lambda_{min},\lambda_{max}]$, which should satisfy $\lambda^* \in I_\lambda$. We can choose $\lambda_{min} = 0$. For the upper bound, two cases should be considered.
If $t_{1n} > T_{max}$, then the problem is infeasible since even the direct trip from $1$ to $n$ is infeasible.
Else, if $t_{1n} \leq T_{max}$, the direct trip from $1$ to $n$ is a feasible solution, and with the inventory decisions $y_1=0$, $y_n=0$, the following inequality holds:
\begin{align}
L(\lambda) \leq c_{1n} - \lambda (T_{max}  - t_{1n}) \label{f:LBineq1} 
\end{align}

Since (\ref{f:LBineq1}) is valid for $\lambda^*$ and if $t_{1n} > T_{max}$ we can find the following bound: 
\begin{align}
\lambda^* \leq \frac{c_{1n} -L(\lambda^*)}{(T_{max}  - t_{1n})} \leq \frac{c_{1n} - L(0)}{(T_{max}  - t_{1n})} \label{f:LBineq2} 
\end{align}

Therefore $\lambda^*$ is contained in the following interval:
\begin{align}
\lambda^* \in \left[0,\frac{c_{1n} - L(0)}{(T_{max}  - t_{1n})}\right] \label{f:LBinterval} 
\end{align}

We performed extensive computational experiments to find a most adequate search algorithm for~$\lambda^*$. Using the intersection of sub-gradients, at each iteration, as next search point turned out to be an efficient and robust choice, as it led to the average least number of evaluations of $L(\lambda)$.
Thus, for a current search interval $[\lambda_a,\lambda_c]$, the next Lagrange multiplier is obtained as the intersection of the lines originating in $\lambda_a$ and $\lambda_c$, with slopes  $-s(x_{\lambda_a})$ and $-s(x_{\lambda_c})$, respectively. The method iterates until the termination is satisfied, in our case the difference $\lambda_c - \lambda_a < \epsilon$ for a given $\epsilon$.  
The overall method is reported in Algorithm~\ref{alg:bsec}.

\begin{algorithm}
\caption{Algorithm for solving the Lagrangian dual} 
\begin{algorithmic}\label{alg:bsec}
\STATE Input: $\lambda_a <  \lambda_c$ and $x_{\lambda_a}$, $x_{\lambda_c}$ such that $s(x_{\lambda_a})>0 $ and $s(x_{\lambda_c})<0$
\STATE Output: feasible solution that meets the termination criterion
\WHILE{termination criterion not reached}
\STATE $\lambda_b \gets \Theta(\lambda_a,\lambda_c,s(x_{\lambda_a}),s(x_{\lambda_c}))$
\STATE calculate $x_{\lambda_b}$ and $s(x_{\lambda_b})$
\IF{$x_{\lambda_b}$ optimal and feasible}
\STATE return $(\lambda_b, x_{\lambda_b})$
\ENDIF
	\IF{$s(x_{\lambda_b}) > 0$}
	\STATE $\lambda_c \gets \lambda_b$
	\ELSE[$x_{\lambda_b}$ infeasible]
	\STATE $\lambda_a \gets \lambda_b$
	\ENDIF
\ENDWHILE
\STATE return $(\lambda_c, x_{\lambda_c})$
\end{algorithmic}
\end{algorithm}

In this algorithm, the function $\Theta(\lambda_a,\lambda_c,s(x_{\lambda_a}),s(x_{\lambda_c}))$ returns the new search point~$\lambda_b$, at each iteration. The search intervals $[\lambda_a,\lambda_c]$ are such that the primal solution that correspond to the left border is always feasible, and the solution that corresponds to the right border is infeasible or optimal.
After termination, the primal solution that corresponds to the left border is feasible, and therefore can be used as an upper bound. 

Naturally, this primal solution is not necessarily optimal for the original ARELTP, since a duality gap can exist. To solve the ARELTP exactly, we rely on this Lagrangian relaxation to provide good combinatorial lower (and upper) bounds within a branch-and-bound framework. The details of this algorithm are described in the next section.

\subsection{Branch and bound}
\label{sec:BB}

The search space of the proposed branch and bound consists of ARELTP solutions \mbox{(\ref{ARELTP:FB}--\ref{ARELTP:y})} where the duration constraint is not necessarily satisfied. The solution branches of the branch-and-bound tree explore subspaces that are restricted by excluding customers and also by defining mandatory customers. We thus represented a solution branch as a pair of sets $(S,T)$, where $S$ the set of mandatory customers and $T$ is the set of excluded customers. Solutions that satisfy the duration constraint are called feasible.

There are four determining design choices; the choice of lower bounds, the construction of (good) feasible solutions (UBs), the branching rule  and the strategy for selecting the next branch. A general discussion of branch and bound strategies can be found in \cite{linderoth1999computational} and a review on branching rules can be found in \cite{achterberg2005branching}.

In our case, the lower bounds are produced by the Lagrangian relaxation with respect to the duration constraint. Slight modification to the methods presented in the previous sections are necessary to reflect mandatory and excluded customers. Excluded customers can be simply eliminated from the data set. To consider a mandatory customer $s \in S$, we eliminate the arcs $(i,j)$ that \emph{skip} this customer, i.e., such that $i < s < j$ for $s \in S$.  

The branching rule and the heuristic to obtain good feasible solutions (UBs) with respect to the duration constraint are based on the solution of the Lagrangian Relaxation, i.e. for a given solution branch $(S,T)$ a feasible solution can be obtained from the Lagrangian relaxation. Based on this solution, an augmented tour $\bar{\tau}$ is constructed by sequentially inserting customers as long as it is possible to maintain feasibility. Worsening the solution is allowed at this step.
As  a consequence, the set of solutions $\Theta = \{\tau : \tau  \subset \bar{\tau} \}$ always satisfies the duration constraint, and we can find an optimal solution $\tilde{\tau}$ in $\Theta$ using DP without duration constraint. This solution defines the upper bound for the solution branch.
 
Our branching rule is based on the selection or elimination of a customer visit $i \not\in \bar{\tau} \setminus T$.  As a consequence, the branch $(S ,T)$ is replaced by the branches $(S \cup \{i\},T)$ and $(S ,T \cup \{i\})$ in the tree. The following rules were considered: sequential branching, strong branching with one-step look ahead and random selection. In case of sequential branching, the smallest index was selected. In case of strong branching, the customer $i$ with the smallest upper bound was selected. Strong branching significantly reduced the number of iterations (a reduction of approximately 30\%), but the evaluation of all candidates was time consuming. In average, the performance of sequential branching and random selection were not significantly different, but the occurrence of outliers were less frequent with random selection. Therefore, the random selection has been selected.
The selection of the next branch for exploration is based on ``best-estimate'', i.e., lowest upper bound. 

Finally, we cannot assume that the triangle inequality are satisfied for some lot sizing applications, leading to some necessary adaptations of the algorithms. This is discussed in the next section.

\section{No triangle inequality} \label{sec::triangle}
In general, the triangle inequality is not satisfied for problems that include setup costs or setup times. 
In the lot sizing application LSwRC, skipping a customer (i.e., a production period) can increase the total setup costs as well as the total resource consumption, and render the solution infeasible.
Four adaptations of the code are needed to resolve these~issues:

\begin{itemize}[nosep,leftmargin=*]
\item[--] When applying dynamic programming to solve the Lagrangian Relaxation for a solution branch (Sections  \ref{sec:DPnoTmax} and \ref{sec:BBLR}), the recursion formula considers the cost to reach a customer~$j$ directly after customer $i$ (cost of $c_{ij} + \lambda t_{ij}$). If the triangle inequality is not satisfied, then a shortest path between $i$ and~$j$ is used.

\item[--] 
To find a good UB for a solution branch, our heuristic procedure performs successive customer insertions and checks the feasibility with respect to the duration constraint. If the triangle inequality holds, the insertion of a customer $i$ between two consecutive customers $\tau_j$, $\tau_{j+1}$ results in an increase of the duration,
which cannot not be shortened by intermediate stops. If the triangle inequality is not satisfied, then a shortest path between~$\tau_j$ and $i$ and a shortest path between $i$ and $\tau_{j+1}$ is used.

\item[--] The dynamic programming approach that considers the duration constraint (based on Equation \ref{formulat:DPt2}) considers all feasible predecessors $i$ of $j$.  
In general, the cost matrix $(c_{ij})$ and the duration matrix $(t_{ij})$ do not satisfy the triangle inequality.
Several non-dominated paths may connect $i$ and $j$ with respect to the two objectives: cost and duration. Therefore all these paths need to be considered.

\item[--] The initial route $(1,n)$ is used to check if the problem is infeasible and to obtain initial bounds for $\lambda^*$ (with respect to Equation \ref{f:LBinterval}). Without the triangle inequality, this route can be infeasible, and a shortest path algorithm is used to generate an initial route.
\end{itemize}

\section{Computational Experiments}
 \label{sec:comp}
In the following, we evaluate the performance of the proposed methods on the ARELTP with and without the duration constraint. The corresponding experiments are based on instances for two types of applications, the evaluation of routes for the SRLTP (SRLTP instances) and lot sizing application (LSwRC instances).

The SRLTP instances are based on two types of orienteering instances, the Tsiligirides instances introduced in \cite{tsi84} and the Chao-Golden instances introduced in \cite{chao93,chao96}; where the distance matrix also serves as cost matrix ($c_{ij} = t_{ij}$). For each node, a piecewise linear profit functions with four steps was generated. The duration limit $T_{max}$ is taken from the original instances and $Q_{max}$ takes the values 30, 60 and 120. For each instance, twenty a-priori routes were constructed by a randomized nearest neighbor procedure.  

To the knowledge of the authors, instances with setup times that are dependent on idle times are not available in the literature. In the following section, the generation of the LSwRC instances will be discussed. All instances are available at \url{http://www.univie.ac.at/prolog/research/ARELTP/}.

\subsection{Instances for the LSwRC} \label{sec:InstLSwRC}

For simplicity, the setup related resource is defined as the setup related costs ($t_{ij} = c_{ij}$), therefore $T_{max}$ is the maximum budget for setup related costs.
To investigate the influence of the structure of the instances on the performance, the following factors are considered: the number of periods $n$, the maximum setup cost $T_{max}$ and the capacity bounds $Q_{max}$.

The production costs are piecewise linear and concave with three steps. Each parameter setting for the random generator is replicated with eleven different random seeds. The number of periods is given by $n \in \{10,20,30,40,50\}$ and $Q_{max} \in \{10,50,100\}$. The tightness of the maximum setup cost constraint $T_{max}$ is strongly dependent on the setup costs $t_{ij}$ and the number of periods $n$. Therefore, we used the following formula for $T_{max}$, which considers a parameter $\theta$: 

\begin{equation}
T_{max} =  \max_{ij} t_{ij} + \theta \underset{SC_{L4L}:=}{\underbrace{ \sum_i t_{i,i+1}.}}
\end{equation}

As such, $T_{max}$ is the sum of the maximum setup cost and a portion of the value $\sum_i t_{i,i+1}$.
This value is the setup cost of the lot-for-lot policy ($L4L$), where the demand of each period is satisfied by the production of this period directly (zero inventory).

If the triangle inequality holds for $t_{i,j}$, then $\theta=1$ means that $T_{max}$ is larger than the setup cost for any feasible solution. The results will be reported with respect to the tightness of $T_{max}$, i.e., ``small'' will be used for $\theta=\frac{2}{5}$, ``medium'' will be used for $\theta=\frac{3}{5}$ and ``large'' will be used for $\theta=\frac{4}{5}$. Analogously, the terms small, medium and large will be used for $Q_{max} \in \{10,50,100\}$.

To evaluate the balance of the different cost components in the instances, we compare the value of each  optimal solution $Z^*$ with the value of the lot-for-lot strategy $Z^{L4L}$.
The cost $Z^{L4L}$ is the sum of the production costs $\sum_i f_i(d_i)$ plus setup costs $\sum_{i=1}^{n-1} c_{i,i+1}$. For an optimal solution $(x^*,y^*,q^*)$, the relative savings $\delta Z$ with respect to the lot-for-lot solution is calculated as:
\begin{align}
\delta Z &= \frac{Z^* - Z^{L4L}}{Z^\textsc{L4L}}  \nonumber \\
&=   \underset{\delta Z^\textsc{prod}}{\underbrace{\displaystyle \frac{1}{Z^\textsc{L4L}} \sum_{i=1}^n \left( f_i(d_i) - f_i(y^*_i) \right) }} - \underset{\delta Z^\textsc{setup}}{\underbrace{ \frac{1}{Z^\textsc{L4L}} \sum_{i=1}^{n-1}  \left( \displaystyle \sum_{j=i+1}^n  c_{ij}x^*_{ij} - c_{i,i+1} \right) }} - \underset{\delta Z^\textsc{inv}}{\underbrace{\displaystyle  \frac{1}{Z^\textsc{L4L}} \sum_{i=1}^n h_i q^*_i}} \nonumber
\end{align}
 
The relative savings $\delta Z$ is the sum of the production savings $\delta Z^\textsc{prod}$ minus the additional expenses due to setups $\delta Z^\textsc{setup}$ and inventory $\delta Z^\textsc{inv}$ in the optimized solution. Figure \ref{statRelCosts} reports the average value of these savings for all instance types, for different values of~$T_{max}$ and~$Q_{max}$. 

The value of the savings tends to increase with the capacity and duration limits~$Q_{max}$ and~$T_{max}$. Indeed, larger capacities result in more opportunities to reduce production costs, at the cost of additional storage or setup costs. It is also noteworthy that all cost components have a significant impact, such that the instances are well-formed and the optimal solution does not correspond to a simple policy such as L4L.

\begin{figure}[htb]
\hspace*{2.7cm}
\includegraphics[width=0.78\hsize]{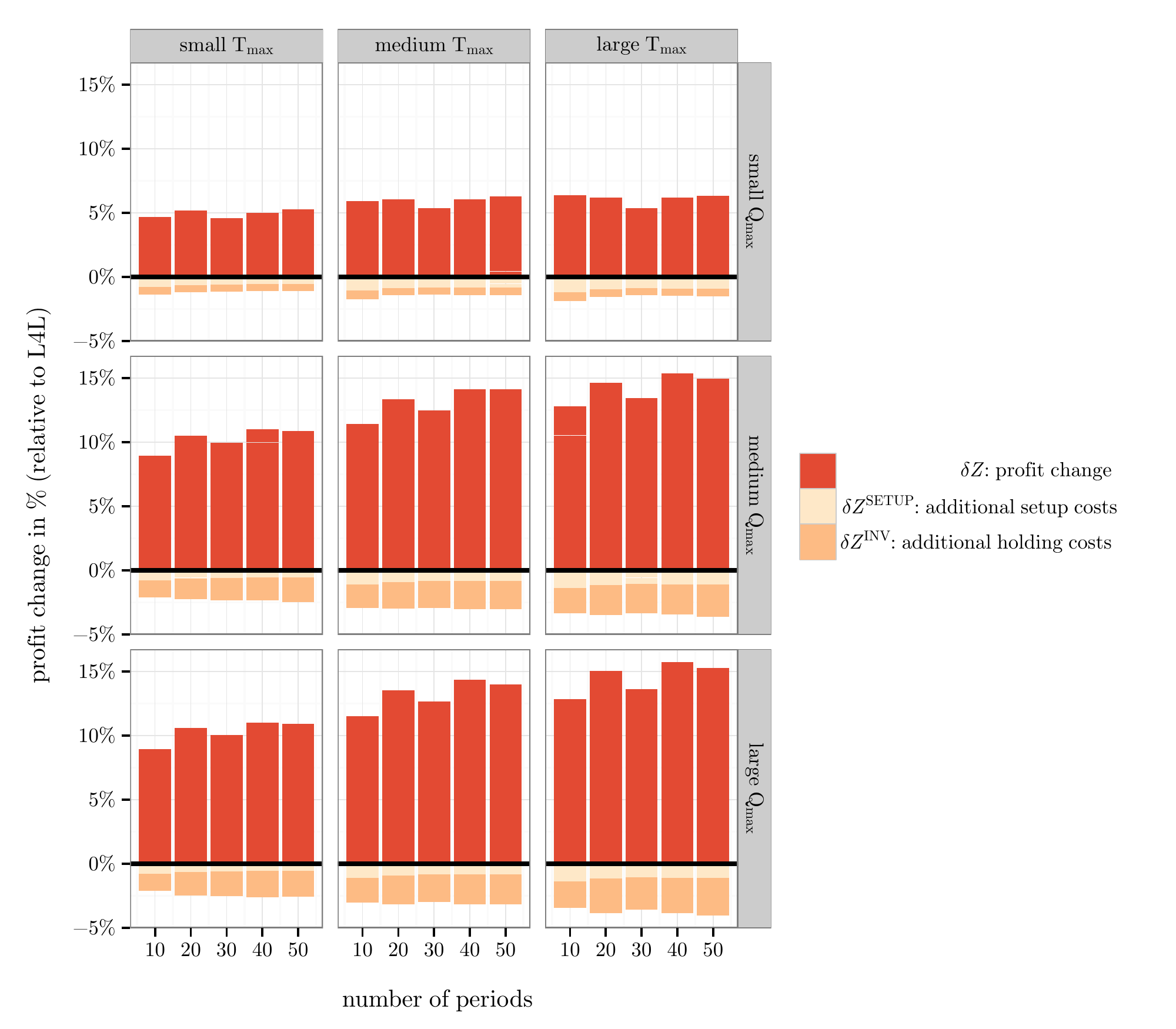}
\caption{Production savings, setup costs and inventory costs when comparing an optimal solution to the L4L policy, for different instance characteristics.}
\label{statRelCosts}
\end{figure}

\subsection{Results when the duration constraint is not considered}

In this section, we first evaluate the performance of the proposed dynamic programming algorithm (Section \ref{sec:DPnoTmax}) for the  ARELTP without duration limit constraint, in comparison with Gurobi  (Section \ref{sec:MIPsolverformualtion}). The experiments were done with the three families of instances: Chao-Golden, Tsiligirides, and the new LSwRC test sets.

For both types of instances the dynamic programming approach is significantly faster than Gurobi.
For the SRLTP and LSwRC, we observe average speedup factors of $50$ and $100$, respectively. 
On $99.96\%$ of instances, DP is faster than Gurobi. The average the speedup factor is $55.9$, with a minimum of $2.66$ and maximum of $284$.
According to a Wilcoxon signed-rank test, the p-value of the null hypothesis (both methods are equally fast) is smaller than $10^{-5}$. This hypothesis can thus be rejected with high confidence, thus validating the significance of this speed-up.

\begin{figure}[htbp]
\centering
\includegraphics[width=0.85\textwidth]{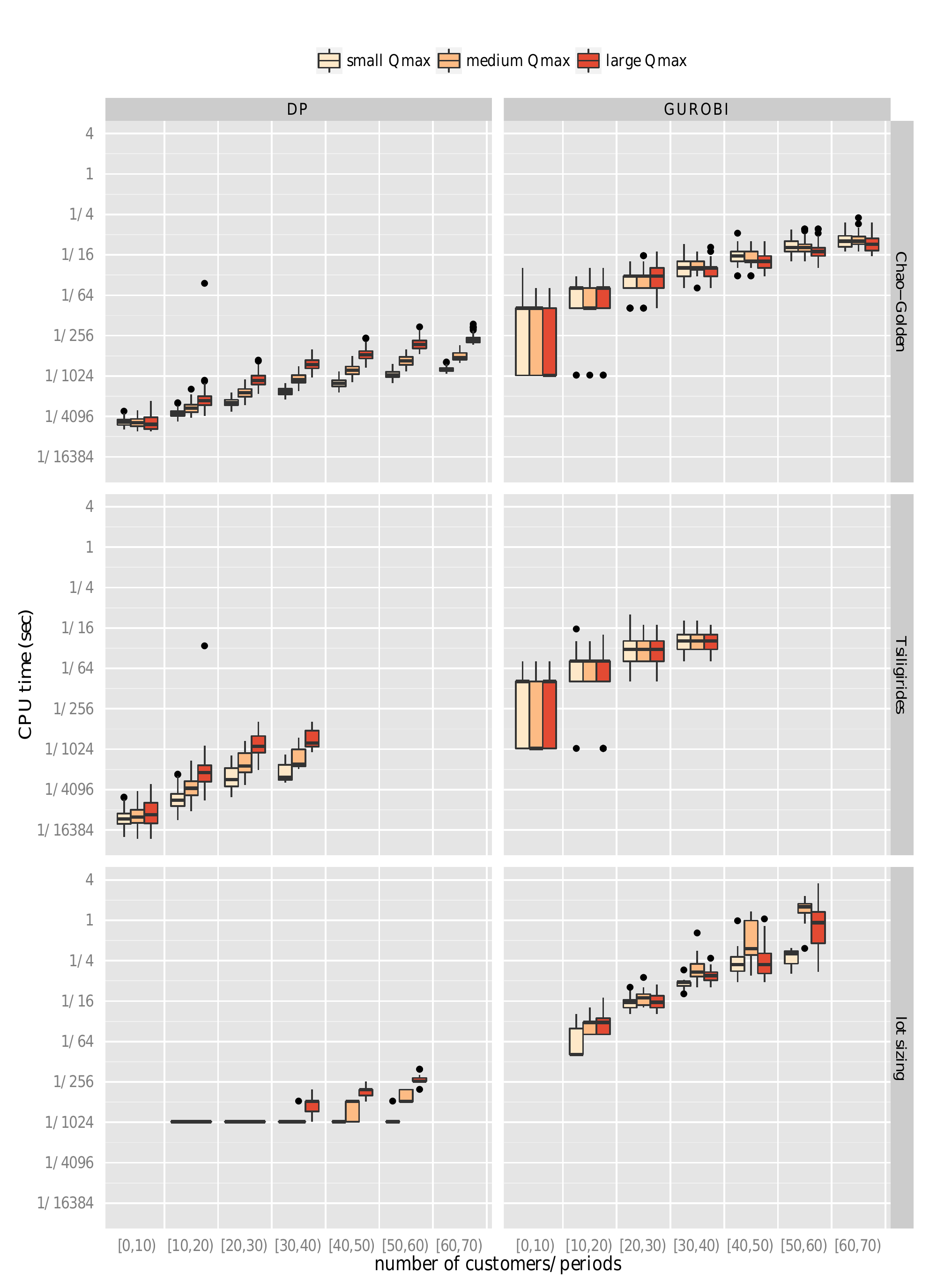}
\caption{CPU time of DP and Gurobi for the ARELTP without duration constraints, reported as a Boxplot for each instance class and each value of $n$ and $Q_{max}$.}
 \label{fig:TSICHAOLOTprocessingtimeDPwoTmax}
\end{figure}

Now, Figure \ref{fig:TSICHAOLOTprocessingtimeDPwoTmax} illustrates the influence of $n$ and $Q_{max}$ on the performance of both methods. We relied on a logarithmic representation of the CPU time as a function of $n$, and display the measures as Boxplots for each instance class and each value of $n$ and $Q_{max}$.
Both algorithms display a CPU time which increases exponentially with problem size.
As expected, the instances with large value of $Q_{max}$ are more difficult for the dynamic programming algorithm, with a CPU time three to four times higher than for small $Q_{max}$. This is due to an increase of the number of labels needed to represent the value functions. In contrast, Gurobi is relatively insensible to changes in the capacity limits.
Overall, DP is significantly faster than Gurobi for all these test instances.

\subsection{Results when considering the maximum duration constraint}

In this section, we compare the performance of the three algorithms for the ARELTP with duration constraints:
\begin{itemize}[nosep]
 \item DP3d -- the proposed pure dynamic programming approach (Section \ref{sec:DPTmax}),
 \item BBDP -- the proposed branch and bound method (Section \ref{sec:BB}) in combination with the DP approach for the Lagrangian relaxation (Section \ref{sec:DPnoTmax}), and
 \item Gurobi -- as a reference method, based on the MIP formulation of Section \ref{sec:MIPsolverformualtion}.
\end{itemize}

For this purpose, we rely on the same three classes of instances. Figure \ref{fig:violin:detail} illustrates the CPU time of each method for each set, and for each value of $Q_{max}$ and $T_{max}$. As it was impossible to display all results separately for each value of $n$ (hence considering four combined factors, for a total of 5633 instances) we aggregated these results and display the shape of the distribution of the CPU time for each sub-class of instances ---with varying $n$--- as violin plots. In these plots, the central bars represent the median and the stars represents the means. The detailed results, for all instances, are available at  \url{http://www.univie.ac.at/prolog/research/ARELTP/}.

\begin{figure}[h]
\centering
\includegraphics[width=0.9\textwidth]{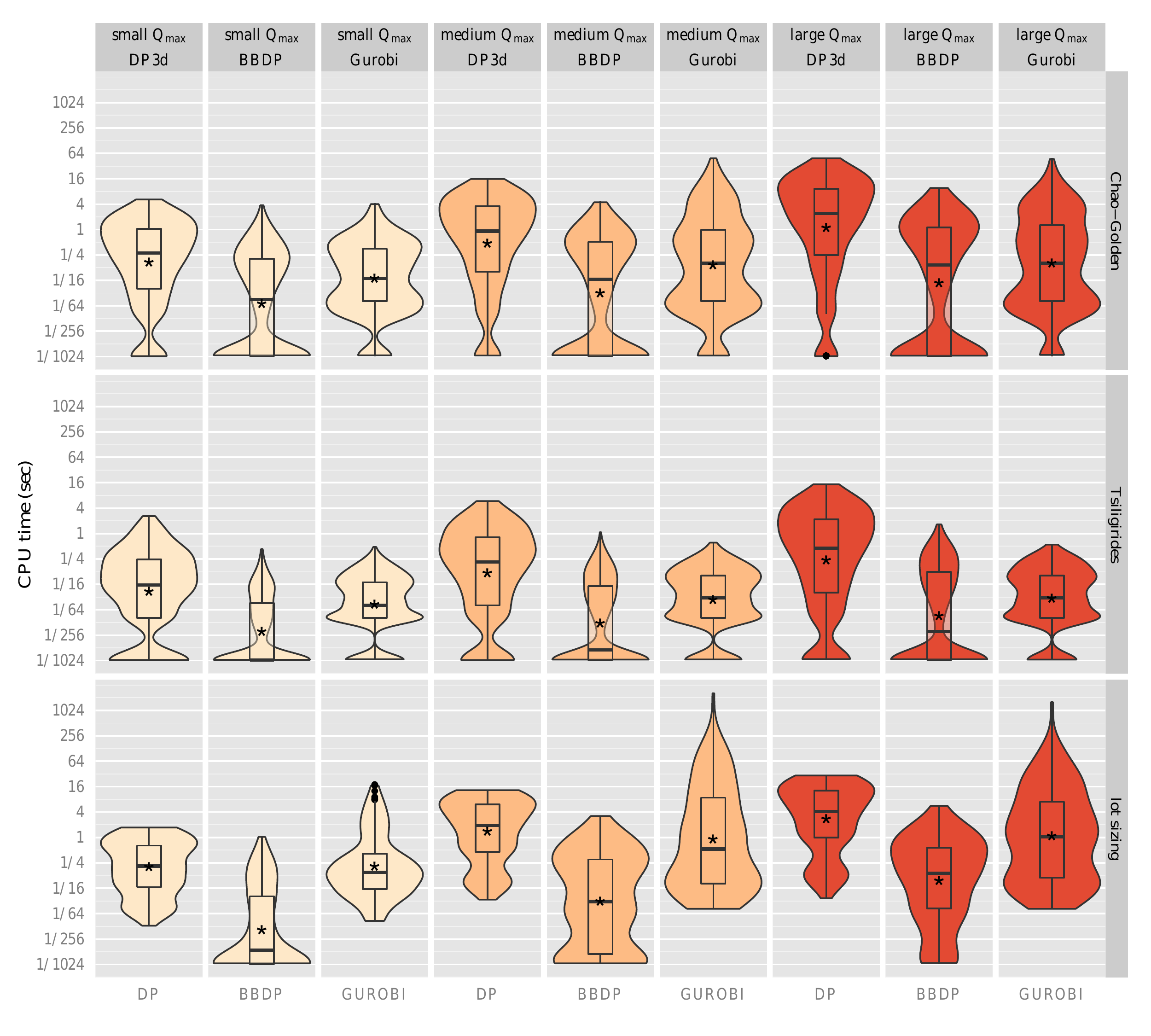}
\caption{CPU time of Gurobi, BBDP and DP3d for each instance class and each value of $Q_{max}$ and $T_{max}$. The results are reported as combined Boxplots and Violin-plots, hence illustrating the distribution of the CPU time each subclass of instances.}
 \label{fig:violin:detail}
\end{figure}

We observe in Figure \ref{fig:violin:detail} that the pure dynamic programming approach (DP3d) is in average slower than the two other methods in the presence of duration constraints, with the exception of some lot sizing instances with medium or large $Q_{max}$. For some of these latter instances, Gurobi requires a long CPU time, as illustrated by the extended tail of the violin plot.
Comparing the branch-and-bound and dynamic programming approach (BBDP) with Gurobi, we observe that BBDP performs significantly better than Gurobi for the three classes of instances, with an average speedup factor of two on the SRLTP instances, and a speedup of ten on the lot sizing instances.
BBDP is also more robust than Gurobi: its CPU time never exceeds ten seconds, while Gurobi uses up to blue thirty minutes to solve some specific instances. This impact is more acute for the lot sizing data sets. We finally observe a large proportion of instances solved in a few milliseconds by BBDP. 

In the following, we further investigate the impact of the problem size $n$ and the capacity~$Q_{max}$ on the solution time. Figure \ref{fig:cpuforDPandBB} displays Boxplot representations of the CPU time of the three methods, for all instance classes.
The same observations can be made: DP3d is slower than Gurobi, which itself is in average slower than BBDP.
Furthermore, larger values of $Q_{max}$ lead to increased CPU time for BBDP and DP3d.
This effect is less marked when considering Gurobi, although a small $Q_{max}$ facilitates the resolution. On the figure, the Boxplots representing the CPU time of Gurobi are almost aligned, on the logarithmic scale, indicating an exponential growth of the resolution time as a function of $n$. This growth tends to be more moderate for BBDP and DP3d, for large values of $n$.

Finally, Figure \ref{fig:boxplot:cpuforDPandBBTmax4lot} illustrates the impact of $T_{max}$ on the CPU time of the three methods, on the lot sizing instances. We observe that the resolution time of DP3d is almost not impacted by the value of this parameter, most likely because $T_{max}$ does not help to eliminate labels until the final stages of the resolution. For Gurobi and BBDP, a medium $T_{max}$ appears to lead to more difficult instances, and a large $T_{max}$ is easier to solve. For BBDP, a large $T_{max}$ helps to reduce the size of the branch-and-bound tree. In particular, in the hypothetical case of $T_{max}=\infty$, the BBDP procedure stops at the root node.

\begin{figure}[p]
\centering
\includegraphics[width=0.95\textwidth]{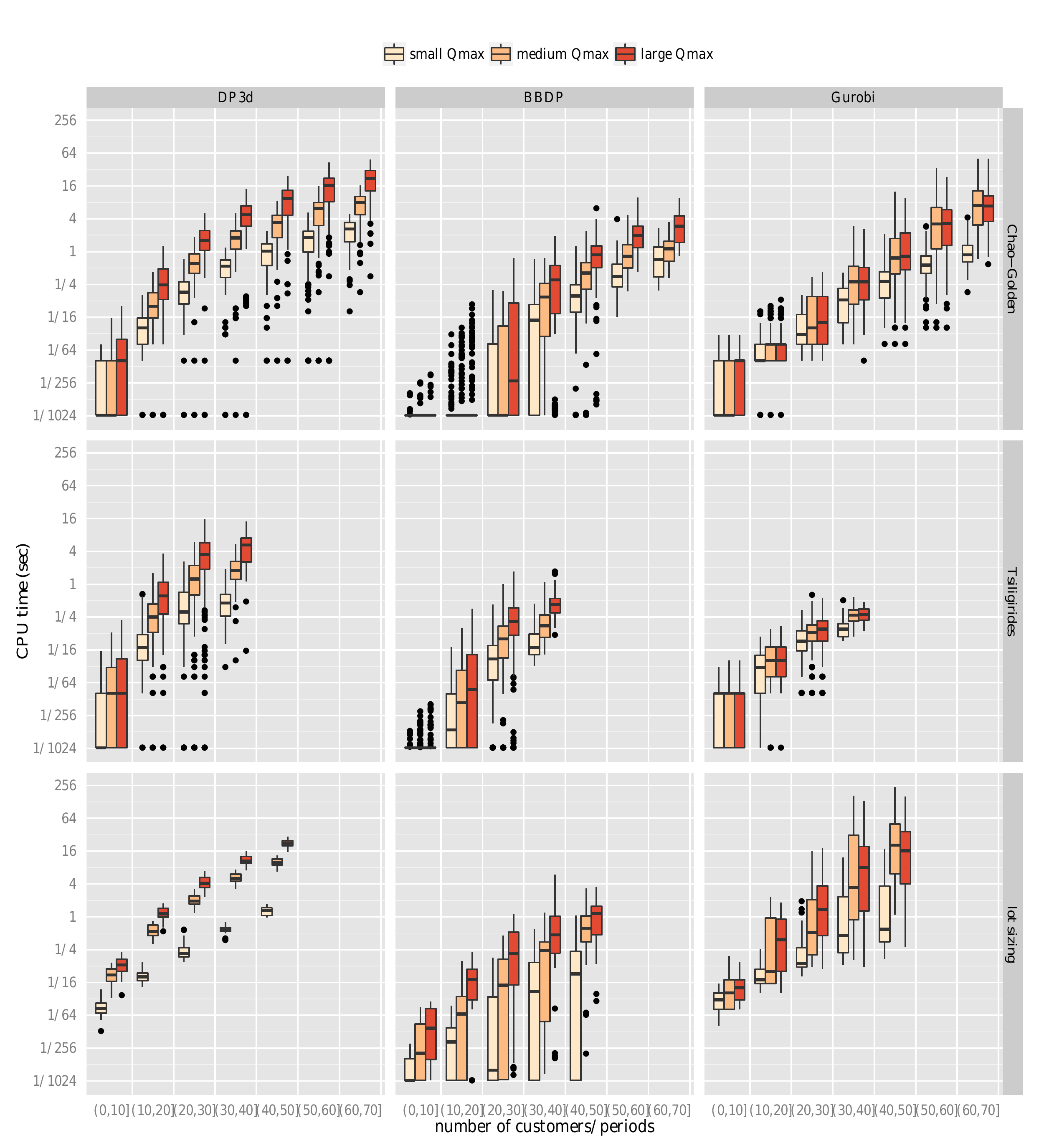}
\caption{CPU time for Gurobi, branch and bound (BBDP) and for the pure DP approach (DP3d) with respect to the number of periods. For Gurobi two lot sizing instances with a CPU time of more than 1000 sec are not documented in the chart.}
\label{fig:cpuforDPandBB}
\end{figure}

\begin{figure}[p]
\centering
\includegraphics[width=0.95\textwidth]{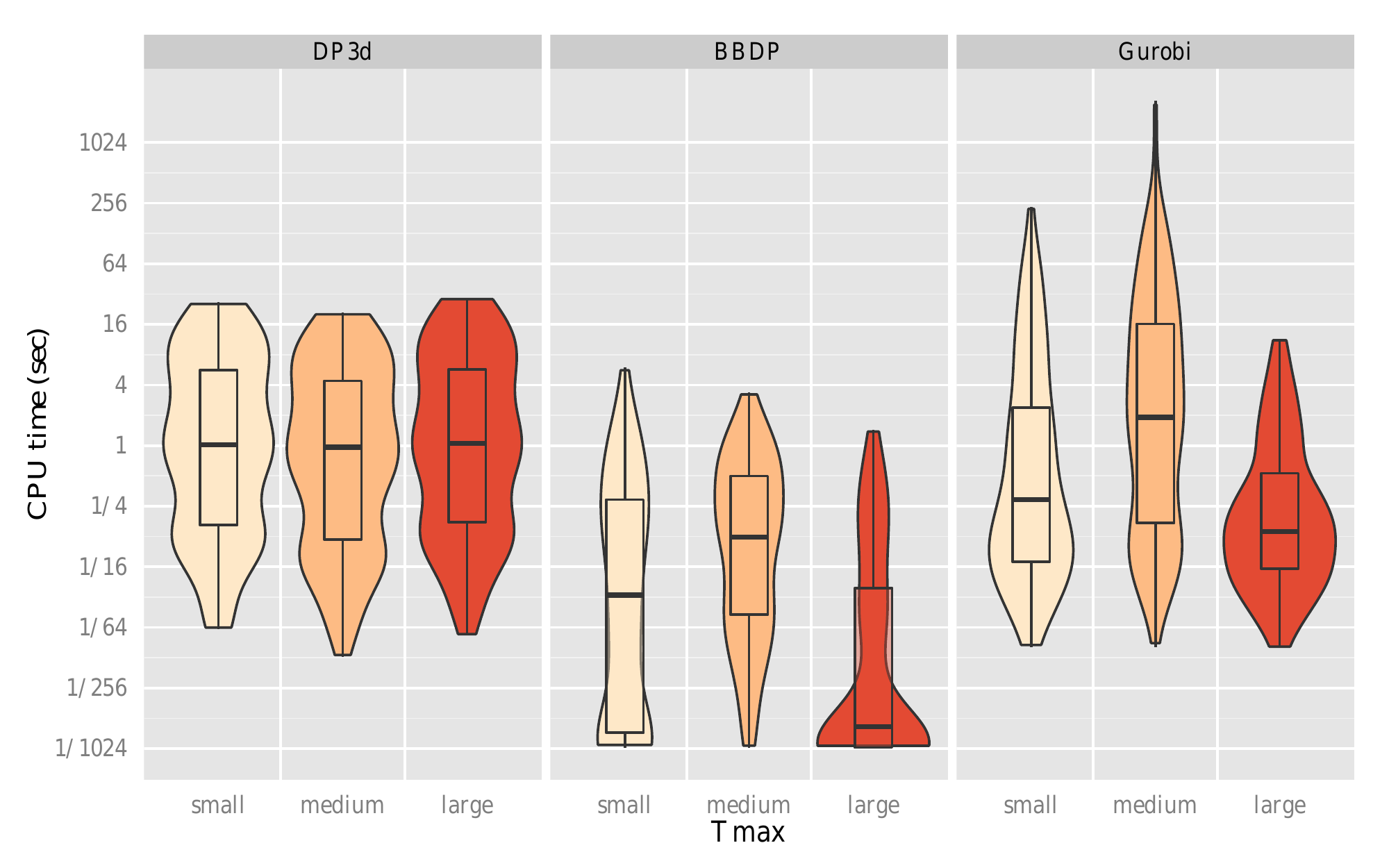}
\caption{CPU time for Gurobi, branch and bound (BBDP) and the pure DP approach (DP3d) for the lot sizing instances with respect to the resource consumption limit $T_{max}$.}
 \label{fig:boxplot:cpuforDPandBBTmax4lot} 
\end{figure}

\clearpage

\section{Conclusions}
\label{sec:conclu}

In this paper, we have introduced the lateral transhipment problem with a-priori routes and piecewise linear profits (ARELTP). We have also shown that this model covers an important class of lot sizing problems with requalification costs. Two problem variants were studied, the ARELTP with and without duration constraint. Both problems are NP-hard. For the ARELTP without duration constraint, we proposed a dynamic programming algorithm (DP) with continuous labels. For the ARELTP with duration constraint, we introduced a generalization of this approach (DP3d) as well as a branch-and-bound method which relies on Lagrangian relaxation and dynamic programming (BBDP).

The performance of the algorithms was evaluated on a large set of problem instances, for the ARELP and for the lot sizing application. Our computational experiments indicate that the proposed DP solves the ARELTP without duration constraint very efficiently, with average speedup factors of $50$ over a commercial solver (Gurobi). For the ARELTP with duration constraint, our hybrid BBDP approach outperforms the commercial MIP solver and the DP3d in most cases, with an average speedup factor of four for the ARELTP instances, and ten for the lot sizing data sets.

These results open the way to several avenues of research. One important application concerns the lateral transhipment problem with one or more routes. We aim to use these exact procedures iteratively, in further works, to evaluate fixed routes and optimize their transhipments during a heuristic resolution. Using the Lagrangian bounds and the proposed hybrid branch-and-bound, it is also possible to estimate if some subproblems (e.g., two new fixed routes evaluated during a local search move) can be eliminated on the fly, before even attempting an exact resolution. In this context, the proposed methods will significantly contribute to solve some difficult multi-attribute vehicle routing problems.  Finally, these methods will be useful to address other complex lot sizing variants and subproblems.

\bibliographystyle{abbrvnat}

\begin{appendices}

\section{Complexity of the ARELTP} \label{sec:complexity}
This section presents the details of the complexity results for the ARELTP (\ref{ARELTP:obj}--\ref{ARELTP:Tmax}). In \cite[][Section 3.3]{DBLP:conf/eurocast/HartlR13} we stated that the ARELTP is solvable in polynomial time if the following three simplifications are premised: 
\begin{enumerate}[nosep]
 \item $f_i$ is linear, \label{cp:1}
 \item the travel costs are not considered ($c_{ij}=0$) and \label{cp:2}
 \item travel time is not considered ($t_{ij}=0$). \label{cp:3}
\end{enumerate}
However, if at least one of these simplifications is revoked, then the ARELTP is NP hard. 
More precisely, the proofs of the following three statements will be provided in this section:
\begin{enumerate}[nosep]
 \item  ARELTP is NP hard if $f_i$ is linear and travel time is not considered (see Lemma \ref{lem:complexity:cij}).
 \item  ARELTP is NP hard if travel costs and travel time are not considered (see Lemma \ref{lem:complexity:flin}).
\item  ARELTP is NP hard if $f_i$ is linear and travel costs are not considered (see Lemma \ref{lem:complexity:tij}).
\end{enumerate}

\begin{lemma} \label{lem:complexity:cij}
The ARELTP is NP hard for $f_i$ linear and $t_{ij}=0$.
\end{lemma}
\begin{proof}
Consider the following reduction of balanced number partitioning (or number partition) \cite{Garey:1979:CIG:578533} for the numbers $A_k$, $k = 1,\ldots,K$ to the ARELTP:
\begin{itemize}[nosep]
\item $n = K+1$
\item $f_1(y_1) = 0$ for $y_1 \in [0,\frac{\sum_{k=1}^K A_k}{2}]$ 
\item $f_{k+1}(y_{k+1}) = y_{k+1}$ for $y_{k+1} \in [-A_k,0]$ 
\item $c_{i,{1+k}} = A_k$ for  $k = 1,\ldots,K$ and  $i = 1,\ldots,k$.
\end{itemize}
The corresponding objective is:
$$ \min_{-A_i<y_i<0} \sum_{i < 1+k} A_k x_{i,1+k} +  \sum_{k} y_k $$ 
And for the transformation $z_k \gets -y_{1+k}$ the objective has the following form:
$$ \max_{0<z_k'<A_k} \sum_{i < 1+k} (z_k - A_k) x_{i,1+k} $$
If $\Lambda \subset \{1,\ldots,K\}$ is a solution to the number partitioning problem,
then $z_k = A_k$ for $k \in \Lambda$ and zero otherwise, therefore a solution for number partition is optimal. On the other hand an optimal solution of this problem with objective zero corresponds to a solution for the number partition problem.
\end{proof}

\begin{lemma} \label{lem:complexity:flin}
The ARELTP is NP hard for $c_{ij}=0$ and $t_{ij}=0$.
\end{lemma}
\begin{proof}
A reduction of the knapsack problem ($v_k,w_k,W$) to the ARELTP will be used in the proof.
Node $1$ is a source with $y_1 \in [0,W]$ and $f_1$ is identical to zero.  
The following cost change functions $f_{1+k}$ are proposed for the transformation:
$$f_{1+k}(y_{1+k}) = \begin{cases} 
	0 \text{, for } y_{1+k} > -w_k  \\
	-v_k \text{, else }
\end{cases}
$$
Then, the corresponding ARELTP has the following equivalent form:
\begin{align}
\min & \sum_{y_{1+k} \leq w_k} - v_{k} \\
 s.t. \,	& y_1 + \sum_{j \leq k} y_{1+j} \geq 0\\
& 0 \leq y_1 \leq W \\ 
	& y_{1+k} \leq 0 
\end{align}
The transformation $z_k \gets -y_{1+k}$ results in the following:
\begin{align}
\max & \sum_{ z_{k} \geq w_k} v_{k} \\
 s.t.	\,&\sum_{k} z_{k} \leq y_1 \leq W \\
  &  z_{k} \geq 0 
\end{align}
Since this problem is equivalent to the knapsack problem the proof is complete.
\end{proof}
A simple modifications of the proof can be used to prove that ARELTP is NP hard for $c_{ij}=0$ and $t_{ij}=0$, and $f_{i}$ continuous.
 
\begin{lemma} \label{lem:complexity:tij}
The ARELTP is NP hard for $f_i$ linear and $c_{ij}=0$.
\end{lemma}
\begin{proof}
Again, a reduction of the knapsack problem ($v_k,w_k,W$) to the ARELTP is possible. Set $T_{max} = W$ and 
$t_{i,k+1} = w_k$; And let $f_1=0$ for $y_1 \in [0,n]$ and $f_{k+1}(y_{1+k}) = v_k y_{k+1}$ for $(y_{k+1} \in [-1,0])$ then the ARELTP has the following form:
\begin{align}
\min & \sum_{ k} y_{1+k} v_{k} \\
 s.t. \,	& y_1 + \sum_{j \leq k} y_{1+j} \geq 0\\
& \sum_{y_{1+k} < 0} w_{k} \leq W \\
& 0 \leq y_1 \leq n \\ 
	& -1 \leq -y_{1+k} \leq 0 
\end{align}
The transformation $z_{k} = -y_{1+k}$ results in:
\begin{align}
\max & \sum_{ k} z_{k} v_{k} \\
 s.t. \,	& \sum_{k} z_{k} \leq  y_1 \leq n\\
& \sum_{z_{k} > 0} w_{k} \leq W  \label{f:comp:W}\\
	& 0 \leq z_{k} \leq 1 
\end{align}
Because of (\ref{f:comp:W}) the weight $w_k$ is activated for $z_k>0$ and therefore $z_k=1$ is the optimal choice when maximizing $\sum_{k} z_{k} v_{k}$; hence $z_{k}$ is binary and the problem is equivalent to the knapsack problem.
\end{proof}

\section{Details on the DP approach for the case where the duration constraint is not considered}

The DP approach for the ARELTP is based on two basic operations; the envelope of functions and the superposition. The calculation of the envelope will be described in Section \ref{appendix:envelope} and the superposition of piecewise linear functions as defined in the main paper (\ref{eq:superos}) will be described in Section \ref{appendix:superpos}. The relatedness of feasible states in dynamic programming and Minkowski sums will be established and cases were the performance can be improved will be discussed. 
In Section \ref{appendix:complexityDP}, the complexity for resolving one stage of the dynamic program is discussed and a reduction of a sorting problem to calculating envelopes indicates that the complexity of the proposed Algorithm is reasonable.  

\subsection{Envelope for piecewise linear functions} \label{appendix:envelope}
The envelope is the minimum of a finite set of functions $\mathcal{G}=\{g_1,g_2,\ldots g_N\}$ and the calculation is accomplished in a divide an conquer manner. In the lowest level, 
the set of functions is decomposed into $\lfloor\frac{N}{2}\rfloor$ pairs plus at most one functions. For each pair the envelope is calculated and replaces the corresponding pair. Therefore, there are at most $\lceil\frac{N}{2}\rceil$ functions in the next level.
According to that, the number of functions in level $l=\lceil\log(N)\rceil+1$ is one, and the corresponding function is the envelope of $\mathcal{G}$. \\

Therefore, to apply this method for a set of piecewise linear functions it is sufficient define the envelope of two piecewise linear functions $f$ and $g$. The union of the domains of $f$ and $g$ can be decomposed into intervals such that each interval has exactly one of the following properties:
\begin{enumerate}[nosep]
\item $f$ is defined, but not $g$ 
\item $g$ is defined, but not $f$ 
\item Both functions are defined
\end{enumerate} 
Such a decomposition can be found by simply traversing the segments of $f$ and $g$ in $\mathcal{O}(n+m)$ and the envelope for each interval is computable in constant time. \\

In each level $l$, the envelopes of pairs of piecewise linear functions is calculated. Based on \cite{agarwal2000davenport, Agarwal:1989:SUL:70501.70506,huttenlocher1992dynamic}, an upper bound for the number of segments in each level and a statement about the runtime complexity can be derived.

More precisely, it is possible to employ results for Davenport--Schnizel sequences to get an upper bound on the number of edges. For a review on Davenport--Schnizel Sequences, see \cite{agarwal2000davenport}. The main results for the continuous case can be found in \cite{Agarwal:1989:SUL:70501.70506} and results for calculating the the envelope of a family of piecewise linear functions are taken from \cite{huttenlocher1992dynamic}, stating that the number of segments is bounded by $\mathcal{O}(M\alpha(M))$ where $\alpha$ denotes the inverse Ackermann function and $M$ is the total number of segments in the first level. It is important to note that this result is only dependent on the number of edges and not dependent on the number of involved functions. 

Therefore the number of edges in each level is also bounded by $M\alpha(M)$, and for estimating the worst case complexity, the calculation of the levels are considered as decoupled. More precisely, for level $l$ let $N_l$ be the number of involved functions. 
To calculate the functions for the next level, the envelope of $\frac{N_l}{2}$ pairs is computed.
For each pair the complexity is linear with the number of involved edges, therefore the computational complexity for each level is bounded by $\mathcal{O}(M\alpha(M))$, with makes $O(M\alpha(M)\log(N))$ in total for all levels.

\subsection{Superposition for piecewise linear functions} \label{appendix:superpos}

A procedure to calculate the superposition $V \boxplus f$ will be presented. The functions $V$ and $f$ are represented by sequences of segments, where each segment is a linear function defined on an interval.
Throughout the paper, we assume that the profit functions are lower semicontinuous. Since $f_i$ are piecewise linear the following is true:
\begin{equation}
 \liminf_{y\to y_0} f_i(y) = f(y_0) \nonumber
\end{equation}

Extending the domains of the segments to closed intervals does not lead to ambiguities, since dominated values do not occur in the envelope. 
Therefore the following notation will be used; the value function $V$ is represented by the segments $v_l^V: v_l(q) = k_l^V q + d_l^V$ for $q \in [a_{l-1}^V,a_{l}^V]$; and the revenue change function $f$ is represented by the segments $v_{l}^f: v_{l}(y) = k_{l}^f y + d_{l}^f$ for $y\in [a_{l-1}^f$ and $a_{l}^f]$. 

In order to calculate the superposition of $V$ and $f$, the notion of feasible state and Minkowski addition will be used. A feasible state with respect to $V$ and $f$ is a point $(q,z)$ such that $q-y \in \mathcal{D}(V)$,  $y \in \mathcal{D}(f)$ and $z = V(q-y) + f(y)$. The set of all feasible states is called $\Gamma_{V,f}$ and the superposition is the envelope of $\Gamma_{V,f}$, or $\Gamma$ for short. This set can also expressed by the envelope of a Minkowski sum:  
$$ \Gamma = \left\{\begin{pmatrix} q \\ V(q) \end{pmatrix}: q \in \mathcal{D}(V)\right\}+\left\{\begin{pmatrix} y \\ f(y) \end{pmatrix}:y \in \mathcal{D}(f)\right\} $$

Efficient algorithms to compute the border of Minkowski sums are known for polygons (cf. \cite{agarwal2002polygon},\cite{ramkumar1996algorithm}) and they can be applied to calculate the superposition. However, the polygons used to represent piecewise linear functions have specific properties, and the worst case complexity can be improved. More precisely, let $|V|$ be the number of segments in $V$ and let $|f|$ be the number of segments in $f$, then $\mathcal{O}(|V||f|\log(|V||f|))$ can be improved to $\mathcal{O}(|V||f|\log(|f|))$.
 
In order to present the procedure, the Minkowski sum will be decomposed using the segments of $V$ and $f$, therefore  $V \boxplus f$ can be represented in the following way: 
\begin{equation}
V \boxplus f  = \min_{l,l'} (v_l^V \boxplus v_{l'}^f) = \text{env} \left( \bigcup_{l,l'}\Gamma_{l,l'} \right)\text{, where } \Gamma_{l,l'} = v_l^V + v_{l'}^f   \nonumber
\end{equation}

In the following, a representation of the set of feasible states $ \Gamma_{l,l'}$ by corner points is established in Lemma \ref{lem2}. It is the basis for the calculation of the superposition for piecewise linear functions, and it will be used to prove certain properties of the value function.

\begin{lemma} \label{lem2}
The set of feasible states $\Gamma_{l,l'}$ is the convex hull of the corner points $P_1^{(l,l')}$, $P_2^{(l,l')}$, $P_3^{(l,l')}$ and $P_4^{(l,l')}$:
$$ P_1^{(l,l')} = \begin{pmatrix} a_{l-1}^V + a_{l'-1}^f \\ v_l^V(a_{l-1}^V) - v_{l'}^f(a_{l'-1}^f)\end{pmatrix} \quad P_2^{(l,l')} = \begin{pmatrix} a_{l}^V + a_{l'-1}^f \\ v_l^V(a_{l}^V) - v_{l'}^f(a_{l'-1}^f) \end{pmatrix}  $$ 
$$ P_3^{(l,l')} = \begin{pmatrix} a_{l-1}^V + a_{l'}^f \\ v_l^V(a_{l-1}^V) - v_{l'}^f(a_{l'}^f) \end{pmatrix} \quad P_4^{(l,l')} = \begin{pmatrix} a_{l}^V + a_{l'}^f \\ v_l^V(a_{l}^V) - v_{l'}^f(a_{l'}^f)\end{pmatrix}$$
\end{lemma}
\begin{proof}
$$ \Gamma_{l,l'} = \left\{ \begin{pmatrix} 0 \\ d_l^V -  d_{l'}^f \end{pmatrix} + q \begin{pmatrix} 1 \\ k_{l}^V \end{pmatrix}  - y\begin{pmatrix} 1\\ k_{l'}^f \end{pmatrix} : \quad q \in  [a_{l-1}^V,a_{l}^V] \quad  y \in  [a_{l'-1}^f,a_{l'}^f]  \right\} $$ 
\end{proof}

A consequence of Lemma \ref{lem2} is that the envelope of $\Gamma_{l,l'}$ is the envelope of the following segments:
\begin{align}
\omega_1^{(l,l')} = \overline{P_1^{(l,l')} P_2^{(l,l')}} \quad
\omega_2^{(l,l')} = \overline{P_1^{(l,l')} P_3^{(l,l')}} \label{def:omega12} \\
\omega_3^{(l,l')} = \overline{P_2^{(l,l')} P_4^{(l,l')}} \quad 
\omega_4^{(l,l')} = \overline{P_3^{(l,l')} P_4^{(l,l')}}, \label{def:omega34} 
\end{align}
and if the slope of $\omega_1^{(l,l')}$ is smaller than the slope of $\omega_2^{(l,l')}$ ($k_l < k'_{l'}$) then the minimum of $\Gamma_{l,l'}$ is defined by $\omega_1^{(l,l')}$ followed by $\omega_3^{(l,l')}$, else it is defined by $\omega_2^{(l,l')}$ followed by $\omega_4^{(l,l')}$
 
In the following, a definition of $\Gamma$ and the dominance of feasible states in $\Gamma$ is formulated.

\begin{definition}
For each segment in $f$, the corresponding set of feasible states \mbox{$\Gamma_{l'}$ is defined as}
\begin{equation}
 \Gamma_{l'} = \bigcup_{l} \Gamma_{l,l'}. \label{formula:gammal1} 
\end{equation}
The set $\Gamma_{l'}^+$ contains all feasible states that correspond to the first $l'$ segments of $f$:
\begin{equation}
 \Gamma_{l'}^+ = \bigcup_{k' \leq l'}\Gamma_{k'}. \label{formula:gammal2}
\end{equation}
\end{definition}

Obviously, $\Gamma$ can be constructed by successively adding segments of $f$ by using the following recursion formula: $\Gamma_{l'+1}^+ = \Gamma_{l'}^+ \cup \Gamma_{l'+1} $. According to (\ref{def:omega12}) and (\ref{def:omega34}), all segments of $\Gamma_{l'}$ can be grouped in the following way:
\begin{align}
\scalebox{1.0}{
$
\begin{array}{lllllllll}
\gamma_1^{(l')}: & &\omega_1^{(1,l')},& \omega_1^{(2,l')},&\omega_1^{(3,l')},& \ldots & \omega_1^{(l_{max}^V-1,l')},& \omega_1^{(l_{max},l')},& \omega_3^{(l_{max},l')} \\
\gamma_2^{(l')}: &\omega_2^{(1,l')},& \omega_4^{(1,l')},&\omega_4^{(2,l')},&\omega_4^{(3,l')},& \ldots & \omega_4^{(l_{max}-1,l')},& \omega_4^{(l_{max},l')} \\
\gamma_3^{(l')}: & &\omega_3^{(1,l')},& \omega_3^{(2,l')},&\omega_3^{(3,l')},& \ldots & \omega_3^{(l_{max}^V-1,l')} \\ 
\gamma_4^{(l')}: & & &\omega_2^{(2,l')},& \omega_2^{(3,l')},& \ldots & \omega_2^{(l_{max}-1,l')},& \omega_2^{(l_{max}^V,l')} 
\end{array}
$
} 
\label{f:gammaRep}
\end{align}
Note that $\omega_1^{(l,l')}$ and $\omega_1^{(l+1,l')}$ are defined for consecutive intervals. Therefore $\omega_1^{(1,l')}, \ldots ,\omega_1^{(l_{max}^V,l')}$ can be used to define a function. The last segment of $\gamma_1^{(l')}$ is $\omega_3^{(l_{max}^V,l')}$, and it is defined for an interval that follows the segment $\omega_1^{(l_{max}^V,l')}$. Therefore, $\gamma_1^{(l')}$ defines a function that starts in $P_1^{(1,l')}$ and ends in $P_4^{(l_{max}^V,l')}$. Analogously, $\gamma_2^{(l')}$ defines a function that starts in $P_1^{(1,l')}$ and ends in $P_4^{(l_{max}^V,l')}$. The segments in the groups $\gamma_3^{(l')}$,  $\gamma_4^{(l')}$ are collinear, but there may be gaps or overlapping in the domains.  \\

Now, a definition of dominance with respect to feasible states will be stated: 

\begin{definition}
A point $(q,z)$ (weakly) dominates a point $(q,z')$ if $z \leq z'$. A set of points $A$ dominates a set of points $B$ if for each point in $(q,z') \in B$ a dominating point in $A$ exists, in symbols: $(q,z) \prec (q,z')$ and  $A \prec B$.
\end{definition}

For instance, if $V$ is monotonically increasing, then $\omega_3^{(l,l')}$ dominates $\omega_2^{(l+1,l')}$. In short,  $\omega_3^{(l,l')} \prec \omega_2^{(l+1,l')}$. The following two lemmas will be useful to prove properties of the superposition $V \boxplus f$:

\begin{lemma} \label{lem:Vmon}
If $V$ and $V'$ are monotonically increasing and $\mathcal{D}(V') \subset \mathcal{D}(V)$, then $\min(V,V')$ is monotonically increasing.
\end{lemma}
\begin{proof}
Suppose that $q_1 < q_2$ and $q_1,q_2 \in \mathcal{D}(V) \cap \mathcal{D}(V')$, then:
$$  \min(V,V') (q_1) \leq V(q_1) \leq V(q_2)  \quad \text{and} \quad \min(V,V') (q_1) \leq V'(q_1) \leq V'(q_2).  $$ 
Therefore:
$$ \min(V,V') (q_1) \leq \min(V,V') (q_2)   $$
If $q_1 \in \mathcal{D}(V) \cap \mathcal{D}(V')$ and  $q_2 \in \mathcal{D}(V) \setminus \mathcal{D}(V')$, we get:
$$ \min(V,V') (q_1) \leq V(q_1) \leq V(q_2) = \min(V,V') (q_2). $$ 
\end{proof}

\begin{lemma} \label{lem:Vcont}
If $V$ and $V'$ are continuous and $\mathcal{D}(V') = \mathcal{D}(V)$, then $\min(V,V')$ is also continuous.
\end{lemma}
\begin{proof}
First case: Suppose that $V(q) \neq V'(q)$, then $V(q) > V'(q)$ or $V(q) < V'(q)$. In the first case, the continuous function $V-V'$ is positive for an interval with $q$ in the interior, a for the second case $V-V'$ is negative inn a interval with $q$ in the interior. In both cases $\min(V,V')$ is defined by one of the continuous function and therefore it is continuous. 

Second case: If $V'(q) = V(q)$, then for a given sequence $q_n$ that converges to $q$ the property $\min(V,V')(q_n) = V(q_n)$ or $\min(V,V')(q_n) = V'(q_n)$ is true for points that are infinitely close to $q$.
In both cases $\min(V,V')(q_n) \rightarrow  V(q)$ follows from the continuity of $V$ and $V'$.
\end{proof}

The following Proposition is the main result of this section, it summarizes some properties of the superposition $V \boxplus f$ that are inherited from properties of $V$ and  $f$:
\begin{proposition} \label{prop:Vi1}
for piecewise linear functions $V,f$ the following statements about the superposition $V \boxplus f$ are true:
\begin{enumerate}[label=(A\arabic*),nosep]
\item $V \boxplus f$ is piecewise linear. \label{lem:Vi1-p1}
\item If $V$ and $f$ are monotonically increasing, then $V \boxplus f$ is monotonically increasing. \label{lem:Vi1-p2}
\item If $V$ and $f$ are continuous, then $V \boxplus f$ is continuous. \label{lem:Vi1-p3}
\end{enumerate}
\end{proposition}

\begin{proof}
\strut
\begin{itemize}[leftmargin=*]
\item
\ref{lem:Vi1-p1}: this is a consequence of a representation of $V \boxplus f$ as the envelope of the Minkowski sum of two polygons.

\item
\ref{lem:Vi1-p2}: 
The strategy to prove the monotonicity is to represent the envelope of $\Gamma$ by the minimum of a set of monotonically increasing functions $\mathcal{F}_{\Gamma}$, where each function starts in the same point and each function ends in the same point. The monotonicity of $V \boxplus f$ is then a consequence of Lemma \ref{lem:Vmon}. 

In the following, a representation for $\Gamma_{l'}$ is proposed where a set of $\mathcal{F}_{\Gamma_{l'}}$ is constructed for a single segment of $f$ such that:
$$ \text{env}(\Gamma_{l'}) = \min(  \mathcal{F}_{\Gamma_{l'}}) $$
Before proposing the construction, note that because of the monotonicity of $V$ the segments $\omega_2^{(l+1,l')}$ can be eliminated ($\omega_3^{(l,l')} \prec \omega_2^{(l+1,l')}$), and as a consequence $\gamma^{(l')}_4$ can be eliminated. Hence, according to (\ref{f:gammaRep}), $\Gamma_{l'}$ can be represented by $\gamma^{(l')}_1,\gamma^{(l')}_2$ and $\gamma^{(l')}_3$:
\begin{align}
\scalebox{1.0}{
$
\begin{array}{lllllllll}
\gamma_1^{(l')}: & &\omega_1^{(1,l')},& \omega_1^{(2,l')},&\omega_1^{(3,l')},& \ldots & \omega_1^{(l_{max}^V-1,l')},& \omega_1^{(l_{max},l')},& \omega_3^{(l_{max},l')} \\
\gamma_2^{(l')}: &\omega_2^{(1,l')},& \omega_4^{(1,l')},&\omega_4^{(2,l')},&\omega_4^{(3,l')},& \ldots & \omega_4^{(l_{max}-1,l')},& \omega_4^{(l_{max}^V,l')} \\
\gamma_3^{(l')}: & &\omega_3^{(1,l')},& \omega_3^{(2,l')},&\omega_3^{(3,l')},& \ldots & \omega_3^{(l_{max}^V-1,l')} \\ 
\end{array}
$
}
\end{align}

\begin{figure}[htbp]
	\centering
		\includegraphics[width=0.7\textwidth]{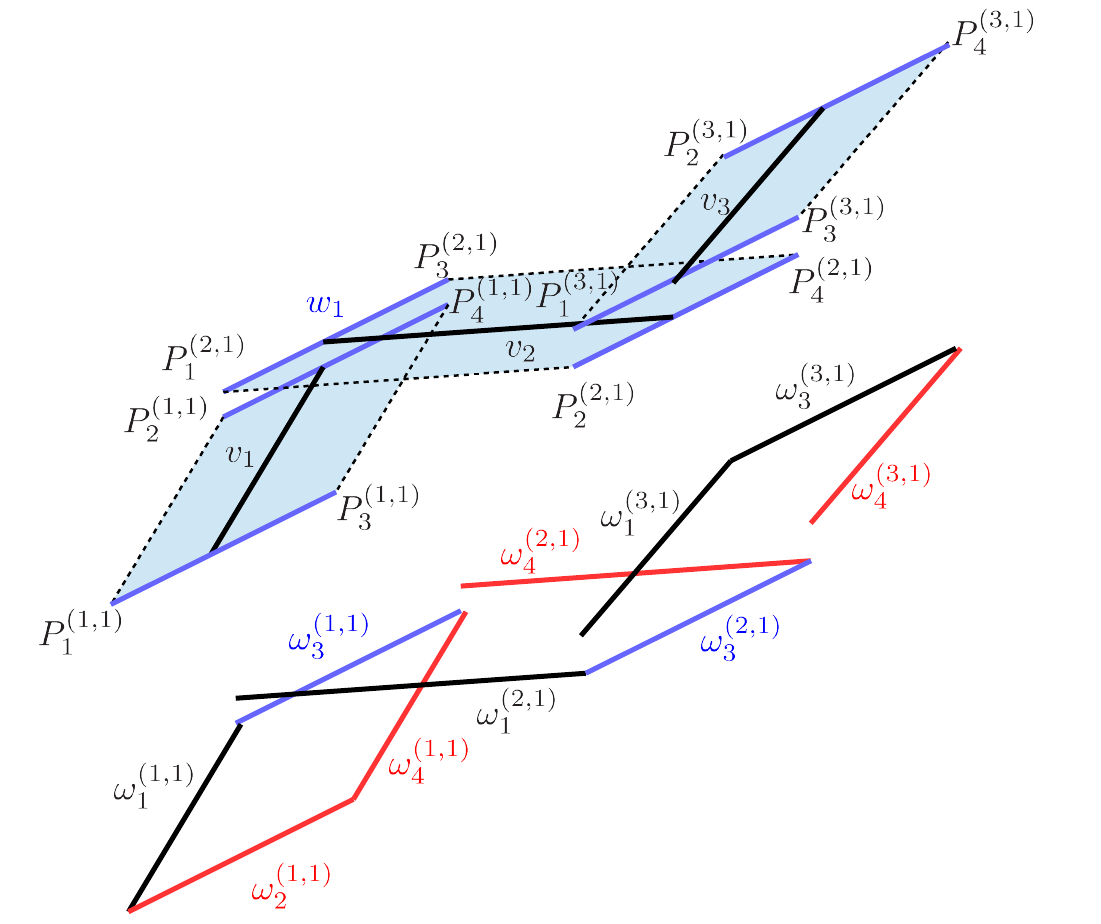}
	\caption{construction of $\gamma_1^{(1)}$(black), $\gamma_2^{(1)}$(red) and $\gamma_3^{(1)}$(blue)}
	\label{fig:superposgamma}
\end{figure}

Figure \ref{fig:superposgamma} demonstrates the construction of $\gamma_1^{(1)}$(black), $\gamma_2^{(1)}$(red) and $\gamma_3^{(1)}$(blue). 
To prove that $\text{env}(\Gamma_{l'})$ is monotonically increasing, it is sufficient to define a set of monotonically increasing functions $\mathcal{F}_{\Gamma_{l'}}$ on the whole domain such that 
each segment of $\gamma_1^{(l')}$, $\gamma_2^{(l')}$ and $\gamma_2^{(l')}$ is used in at least one function;
each function only consists of segments of $\gamma_1^{(l')}$, $\gamma_2^{(l')}$ and $\gamma_2^{(l')}$; and each function starts in $P_1^{(1,l')}$ and ends in $P_4^{(l_{max}^V,l')}$.

We know that $\gamma_1^{(l')}$ and $\gamma_2^{(l')}$ are functions. Because of monotonicity of $f$, they are also monotonically increasing and satisfy the assumptions. To complete the proof, it is sufficient to find monotonically increasing functions on the whole domain that integrate the missing segments in such a way that the function starts in $P_1^{(1,l')}$ and ends in $P_4^{(l_{max}^V,l')}$.
In the following, the auxiliary functions $\eta_3^{(l',l)}$ are defined for each segment $\omega_3^{(l,l')}$:
$$
\eta_3^{(l',l)}: \omega_1^{(1,l')}, \omega_1^{(2,l')},\ldots \omega_1^{(l,l')},\omega_3^{(l,l')},\omega_4^{(l+1,l')}, \ldots  \omega_4^{(l_{max}^V,l')}
$$
Obviously, $\eta_3^{(l',l)}$ contains $\omega_3^{(l,l')}$ and the first part of the sequence of segments $\omega_1^{(1,l')}, \omega_1^{(2,l')}, \allowbreak \ldots, \omega_1^{(l,l')},\omega_3^{(l,l')}$ is monotonically increasing, starting in $P_1^{(1,l')}$. Analogously, the last part of segments $\omega_4^{(l+1,l')}, \ldots  \omega_4^{(l_{max}^V,l')}$ is monotonically increasing and ends in $P_4^{(l_{max}^V,l')}$. Finally the middle piece, $\omega_3^{(l,l')},\omega_4^{(l+1,l')}$ is monotonically increasing because $P_4^{l,l'} \prec P_3^{l+1,l'}$, hence proving that $\eta_3^{(l',l)}$ is monotonically increasing, starting in $P_1^{(1,l')}$ and ending in $P_4^{(l_{max}^V,l')}$.
Now, the envelope of $\Gamma_{l'}$ can be represented by:
$$
\mathcal{F}_{\Gamma_{l'}} = \bigcup_{l'=1}^{l_{max}^V-1} \{\eta_3^{(l',l)}\} \cup \{\gamma_1^{(l')},\gamma_2^{(l')}\}.
$$
Therefore, $V \boxplus w_1$ is monotonically increasing. To prove that $V \boxplus f$ is monotonically increasing for more than one segment, the following auxiliary statement will be used. A set of functions $\mathcal{F}_{\Gamma_{l'}}^+$ exists such that:
\begin{enumerate}[label=(B\arabic*),nosep]
\item $ \text{env}(\Gamma_{l'}^+) = \min( \mathcal{F}_{\Gamma_{l'}}^+) $
\item each functions is monotonically increasing \label{prop:Gammamon}
\item each function starts in $P_1^{(1,1)}$ and ends in $\smash{P_4^{(l_{max}^V,l')}}$. \label{prop:Gammarange}
\end{enumerate}
proof of the auxiliary statement by induction: for $l'=1$ the statement is true, since:
$$\text{env}(\Gamma_{1}^+) = \text{env}(\Gamma_{1})= \min(\mathcal{F}_{\Gamma_{1}}) $$
induction step: $l' \rightarrow l' + 1$. The set of functions $\mathcal{F}_{\Gamma_{l'+1}} = \mathcal{F}_1 \cup \mathcal{F}_2$ consists of two parts. The first part $\mathcal{F}_1$ contains the functions that represent $\text{env}(\Gamma_{l'}^+)$ plus additional segments to meet \ref{prop:Gammarange}. The second set of functions $\mathcal{F}_2$ represents $\text{env}(\Gamma_{l'+1})$ and again additional segments are added to meet \ref{prop:Gammarange}. In the following definitions for $\mathcal{F}_1$ and $\mathcal{F}_2$ the symbol $\oplus$ is used to indicate the concatenation of segments:
$$\mathcal{F}_1 = \mathcal{F}_{\Gamma_{l'}}^+\oplus (\omega_1^{(l_{max}^V,l'+1)},\omega_3^{(l_{max}^V,l'+1)} )$$
$$\mathcal{F}_2 = ( \omega_2^{(1,1)}, \omega_2^{(1,2)}, \ldots \omega_2^{(1,l')}) \oplus \mathcal{F}_{\Gamma_{l'+1}}$$
Finally, the envelope of $\Gamma_{l'+1}^+$ can be represented by the minimum of $\mathcal{F}_1$ and $\mathcal{F}_2$ and since the functions of both sets meet \ref{prop:Gammamon} and \ref{prop:Gammarange} the minimum also satisfies \ref{prop:Gammamon} and \ref{prop:Gammarange}.
$$\text{env}(\Gamma_{l'+1}^+) = \min(\mathcal{F}_1 \cup \mathcal{F}_2)  $$ 

According to the auxiliary statement $\text{env}(\Gamma_{l'_{max}}^+)$ is monotonically increasing and since  $\Gamma_{l'_{max}}^+ = \Gamma$ the proof for the monotonicity is complete.

\item
\ref{lem:Vi1-p3}: the proof is analogous to the proof for property \ref{lem:Vi1-p2} but it is based on Lemma \ref{lem:Vcont} and uses the fact that $P_1^{(l+1,l')}= P_2^{(l,l')}$ and $P_3^{(l+1,l')}= P_4^{(l,l')}$ and therefore $\omega_3^{(l,l')} = \omega_2^{(l+1,l')}$.

\end{itemize}
\end{proof}

\begin{remark}
When calculating the superposition for $f$ or $V$ with discontinuities, the concept of using auxiliary functions (e.g. $\gamma_1$ and $\gamma_2$) can be used for pairs of continuous parts. This helps to reduce the number of involved segments when calculating the envelope of $\Gamma$.
\end{remark}

\begin{remark}
For convex value functions, the calculations can be simplified; for instance if $f$ consist of only one segment, then $\gamma_1$ and $\gamma_2$ may alternate in the envelope at most once and $\gamma_3$ consists of at most one segment. An example can be found in Figure \ref{fig:superposgammaconvex}.

\begin{figure}[htbp]
	\centering
		\includegraphics[width=0.5\textwidth]{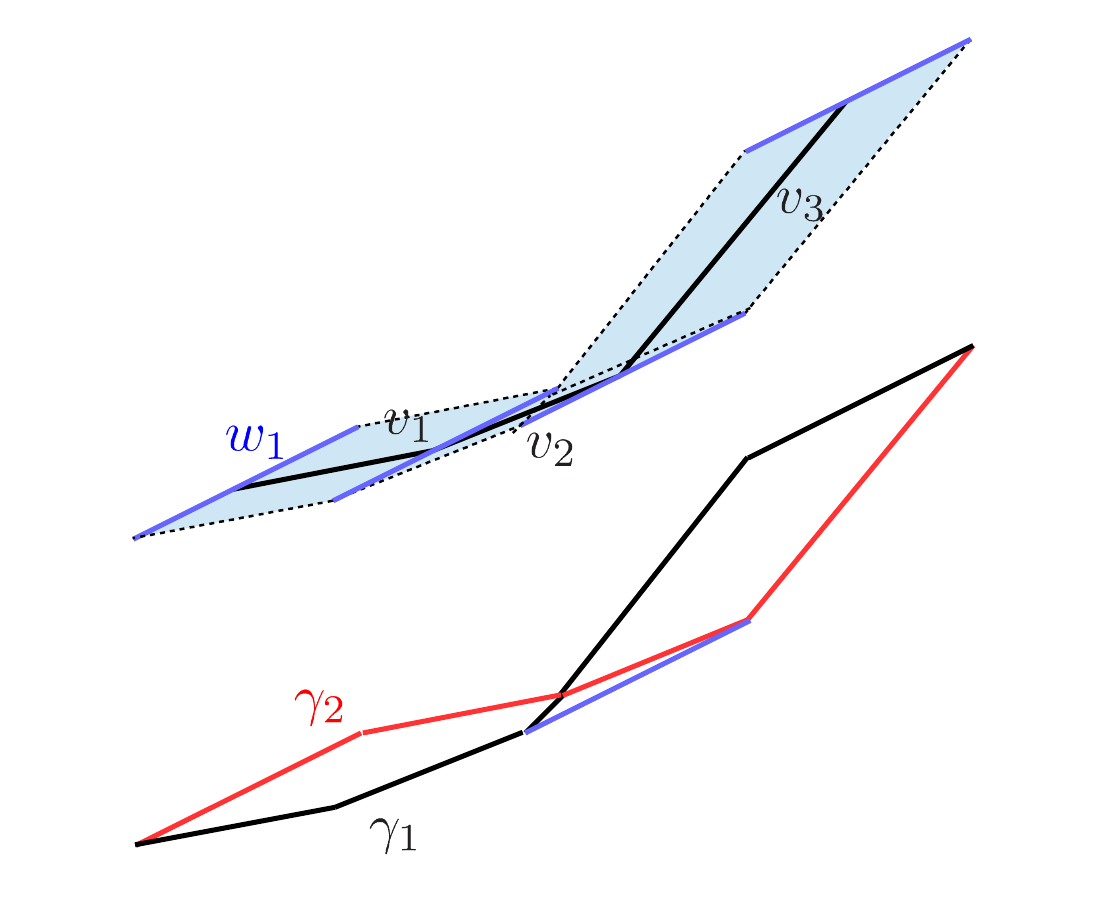}
	\caption{proof - construction of $\gamma_1$ and $\gamma_2$ in the convex case}
	\label{fig:superposgammaconvex}
\end{figure}
\end{remark}

\begin{proposition} \label{prop:superpos:compl}
Given the piecewise linear functions $V,f$,  the superposition $V \boxplus f$ can be calculated in at most $\mathcal{O}(|V| |f| log(|f|))$ time, where $|V|$ is the number of segments in $V$ and $|f|$ is the number of segments in $f$.
\end{proposition}
\begin{proof}
According to the construction given (\ref{f:gammaRep}), it is possible to calculate $\text{env}(\Gamma_{l'})$ in $\mathcal{O}(|V|)$ time. According to that, the calculation of the envelopes for each segment in $f$ can be done in $\mathcal{O}(|V| |f|)$. To calculate the envelope of the envelopes we therefore get a runtime complexity of $\mathcal{O}(|V| |f| log(|f|)$ and the number of edges is bounded by $\mathcal{O}(|V| |f| \alpha(|V| |f|)$.
\end{proof}

\begin{lemma} \label{lem:Vi}
If $f_i$ is monotonically increasing, then $V_{i}$ is monotonically increasing.
\end{lemma}
\begin{proof}
From Proposition \ref{prop:Vi1} and Lemma \ref{lem:Vmon} it is clear that $V_{i}$ is monotonically increasing.
\end{proof}

\begin{example} \label{ex1}
The input data for an example with four locations is given below:
\begin{itemize}[nosep]
\item $f_1(y) = -4y$ for $y \in [0,1]$ 
\item $f_2(y) = -2y$ for $y \in [0,1]$ 
\item $f_3(y) = -4y$ for $y \in [0,2]$
\item $f_4(y) = -5y$ for $y \in [-2,0]$
\item $c_{2,3}=\infty$, else: $c_{i,j}=0$. \\
\end{itemize} 

The resulting value functions are :  
\begin{itemize}
\item $V_1(q) = 4q$ for $q \in [0,1]$ 
\item $V_2(q) = \begin{cases} 2q \text{ for } q \in [0,1] \\
                              -2+4q \text{ for } q \in (1,2]
								\end{cases} $
\item $V_3(q) = 4q \text{ for } q \in [0,3]$
\item $\tilde{V}_4(q) = \begin{cases} 2q \text{ for } q \in [0,1] \\
                              -2+4q \text{ for } q \in (1,2]\\
                               4q \text{ for } q \in (2,3]
								\end{cases} $
\item $V_4(q) = \begin{cases} -4+5q \text{ for } q \in [0,2]\\
                              -3+5q \text{ for } q \in (2,3]
								\end{cases} $
\end{itemize} 

\end{example}

\subsection{Complexity and dynamic programming, without considering the duration constraint} \label{appendix:complexityDP}
We discuss the complexity for resolving a single stage of the dynamic programming recurrence, to calculate $V_i$ in stage $i$ as defined in formula (\ref{f:rec}) in the main paper. The calculations consist of two parts, the first one is the calculation of $\tilde{V}_i = \min_{j < i} (V_j + c_{j,i})$ and the second is the the superposition $\tilde{V}_i \boxplus f_i$, which is equivalent to the calculation of the corresponding envelope $\Gamma$. \\

According to Section \ref{appendix:envelope} the calculation of $\tilde{V}_i$ can be accomplished in a divide and conquer manner, by decomposing the functions $\{ V_j + c_{j,i}\}$ into pairs of functions.
In the present case the number of involved edges is $m = \sum_{j<i} m_j$, where $m_j$ is the number of segments in $V_j + c_{j,i}$. Therefore, $O(m\alpha(m))$ is the complexity for each of the $\left\lceil \log(i)\right\rceil$ levels, which gives worst case complexity of $O(\log(i) m \alpha(m))$ in total. Note that $\alpha(m)$ grows extremely slow and $\alpha(m)<4$ for any practical size of $m$.

To compute the worst case complexity for the superposition, we note that $m\alpha(m)$ can be used as bound for the maximum number of segments in $\tilde{V}_i$. For each segment in $\tilde{V}_i$, the number of edges in $\Gamma$ is bounded the number of pairs of segments $(v,w)$ ($v \in \tilde{V}_i$ and $w \in f_i$). That makes $O(m\alpha(m) K)$, where $K$ is the number of segments in $f_i$.

Overall, the complexity is therefore $O(m \alpha(m)\log(i) +m \alpha(m) K)$. Since $\log(m)$ grows faster than the constant $K$, the complexity for solving the DP recurrence relation for stage $i$  is $O(m \alpha(m)\log(m))$. In order to show that the algorithm has a reasonable complexity, we will reduce a sorting problem to the calculation of the upper envelope of a set of piecewise linear functions, which is the core of the algorithm. This sorting problem is called sorting $M$ sets of $N$ numbers and it is defined as follows:
For $j=1 \ldots M$ the sets $S_j \subset [0,1)$ are given and $|S_j|=N$. The sets $S_j$ are defined by the elements $b_{i,j}$ for $j=1 \ldots M$ and $i=1 \ldots N$, i.e. $S_j = \{b_{i,j}\}$. The problem is to sort the collection $\bigcup_j S_j$.

Before proposing a reduction to the calculation of the envelope, we note that for solving the problem of sorting $N$ numbers there is no algorithm that performs better than $O(N log(N))$ in the comparison model. Therefore we cannot expect for find an algorithm for sorting $M$ sets of $N$ numbers in less than $\mathcal{O}(M N log(N))$. 

\begin{lemma}
Sorting $j=1 \ldots M$ sets $S_j = \{b_{ij} \in [0,1), i=1 \ldots N\}$ of size $N$ can be transformed to calculating the envelope of $N$ functions with $M$ segments.
\end{lemma}
\begin{proof}
The instance is given by the sets $S_j = \{b_{ij}\} \subset [0,1)$, where $j=1 \ldots M$ and $i=1 \ldots N$. The following transformation is used to map all $b_{ij}$ with the same $j$ into the interval $[j-1,j)$:
$$ a_{ij} \gets (j-1) + b_{ij} $$
According to that, sorting $S'_j = \{b_{ij}\}$ is equivalent to sorting $S_j$. For fixed $i$ the sets $L_i$ are defined by $L_i = \{a_{ij}\}$ and the following properties hold: $L_i \subset [0,M)$, $|L_i|=M$ and exactly one member in available in each interval $[j-1,j)$ where $a_{i,j} < a_{i,j+1}$. For the reduction, we consider the following functions:
$$ f_i(y) = \begin{cases} 
    a_{i,1}  &\mbox{if }  y \in [0,a_{i,1}]  \\ 
    a_{i,j}  &\mbox{if } y \in (a_{i,j-1}, a_{i,j}] \quad (j \geq 2)
\end{cases}$$ 
The corresponding functions $f_i$ are monotonically increasing and stepwise constant. An example is given in Figure \ref{fig:reductiontosortedlists}.
The proposed algorithm for calculating the envelope returns the minimum of the collection $\mathcal{F} = \{f_i\}$ with a complexity of $O((M+1)N \log(N) \alpha(N(M+1)) )$ and due to the construction of the envelope, the elements of $\{a_{i,j}\}$ are sorted in $f_{max}$. The maximum function in the interval $[j-1,j)$ returns the sorted values $\{a_{i,j}\}$, hence leading to the order of the original values.
\end{proof}

Sorting $M$ sets of $N$ numbers can be reduced to an envelope computation, currently done in $\mathcal{O}(MN \alpha(MN) log(N))$. Our algorithm for calculating the envelope is thus very close to its theoretical complexity bound. 
 
\begin{figure}
	\centering 
		\includegraphics[width=0.4\textwidth]{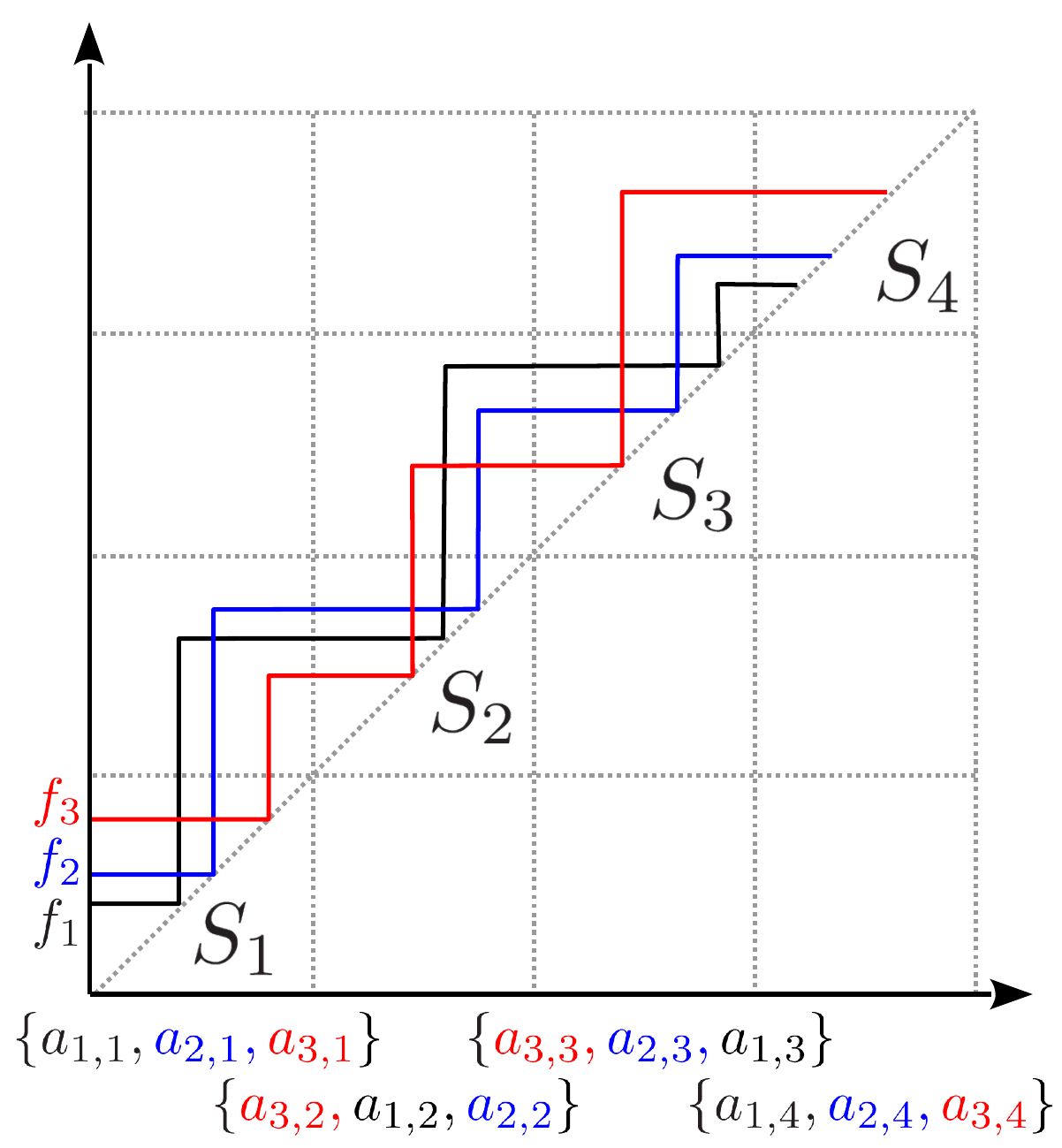}
	\caption{A example for sorting four sets ($S_1,S_2,S_3,S_4$) of size three; the order of the elements appears in the envelope of the corresponding intervals}
	\label{fig:reductiontosortedlists}
\end{figure}

\subsection{Considering integer data} \label{sec:integerdata}

Similar to the total unimodularity for the min cost flow subproblem discussed in \cite{DBLP:conf/eurocast/HartlR13} where integer solutions come for free, it may be useful to consider restricting the domain of the labels of the DP Algorithm to integers.
Suppose that the domains of the segment of $f_i$ have integer borders, then it will be shown that the domains of the segments of the value functions $V_i$ can also be restricted to intervals with integer borders.  
The benefit of this reduction may be explained by the following situation: when calculating the the envelope, tiny segments may appear that do not contain a single integer, therefore the corresponding segment could be eliminated to reduce the complexity.
In the following, an algorithm for restricting the segments of a value function to integer domains will be presented. We show that the superposition of the integerized value functions is equivalent to the superposition of the integerized versions, if the domains of the segments of $f_i$ have integer borders.

Let $\mathbb{Z}$ denote the set of integers, 
 then $W \subset_{\mathbb{Z}} V$ denotes that $W$ is an integerized version of $V$. It means that $W$ is a value function that is identical to $V$ on the domain of $W$, and it includes all integer values of the domain of $V$. More precisely:
\begin{enumerate}[label=(P\arabic*)]
\item $\mathcal{D}(V) \cap  \mathbb{Z} \subset \mathcal{D}(W) \subset \mathcal{D}(V)$   \label{integerize:P3}
\item $ W = V|_{\mathcal{D}(W)} $  \label{integerize:P2}
\item $a_l^W, b_l^W \in \mathbb{Z}$\label{integerize:P1}
\end{enumerate}

\begin{lemma} 
$\subset_{\mathbb{Z}}$ is transitive
\end{lemma}
\begin{proof}
It is sufficient to prove \ref{integerize:P3}. Suppose that $U \subset_{\mathbb{Z}} V$ and $V \subset_{\mathbb{Z}} W$ then $\mathcal{D}(U) \subset \mathcal{D}(V)$ and  $\mathcal{D}(V) \subset \mathcal{D}(W)$ therefore  $\mathcal{D}(U) \subset \mathcal{D}(W)$. It also follows that $\mathcal{D}(W) \cap  \mathbb{Z} \subset \mathcal{D}(V)$ and therefore $\mathcal{D}(W) \cap  \mathbb{Z} \subset \mathcal{D}(V) \cap  \mathbb{Z}$ and in combination with $\mathcal{D}(V) \cap  \mathbb{Z} \subset \mathcal{D}(U)$ it is obvious that $\mathcal{D}(W) \cap  \mathbb{Z} \subset \mathcal{D}(U)$  
\end{proof}

Algorithm \ref{algo:integerize} is used to integerize the border of the domains of the segments of the piecewise linear function $V$ such that $W \subset_{\mathbb{Z}} V$.
In the following we will assume that the domains of $v_l^V$ are closed intervals $I_l^V$. In order to formulate a more general algorithm that allows `holes' in the domain, the interval domains $I_l^V$ are identified by $[a_{l}^V,b_{l}^V]$ where $b_{l}^V \leq a_{l+1}^V$. The output of Algorithm \ref{algo:integerize} is called $V_{\mathbb{Z}}$ and the following Lemma states that Algorithm \ref{algo:integerize} defines a integerized version of $V$.
   
\newcommand{\LINEFOR}[2]{%
    \STATE\algorithmicfor\ {#1}\ \algorithmicdo\ {#2} \algorithmicend\ \algorithmicfor%
}
\newcommand{\LINEIF}[2]{%
    \STATE\algorithmicif\ {#1}\ \algorithmicthen\ {#2} \algorithmicend\ \algorithmicif%
}

\begin{algorithm}
\caption{Create the integer version of a piecewise linear functions $V$ represented by a set of linear functions $v_l$ defined on $[a_{l},b_{l}]$ where $l = 1 \ldots L$.} \label{algo:integerize}
\begin{algorithmic}[1]
    \FOR{$l\in \{1 \ldots L\}$}
		    \STATE $a_l \gets \left\lceil a_{l}\right\rceil $ \label{algo:integerize:1}
				\STATE $b_l \gets \left\lfloor b_l \right\rfloor $ \label{algo:integerize:2}
				\LINEIF{$a_{l} > b_{l}$}{ $V \gets remove(v_l,V)$ }\label{algo:integerize:3}
		\ENDFOR 
\STATE $l \gets 2$
\WHILE{$l \leq  L$}
\IF{$b_{l-1}=a_{l}$} \label{algo:integerize:4a} 
\IF{$v_l(a_{l}) < v_{l+1}(a_{l})$} \label{algo:integerize:4}
\STATE $ b_{l-1} \gets b_{l-1} - 1 $ \label{algo:integerize:5}
\LINEIF{$a_{l-1} > b_{l-1}$}{ $V \gets remove(v_{l-1},V)$ and $continue;$ }\label{algo:integerize:6}
\ENDIF
\IF{$v_l(a_{l}) > v_{l+1}(a_{l})$ } \label{algo:integerize:7}
\STATE $ a_{l} \gets b_{l-1}+ 1 $ \label{algo:integerize:8}
\LINEIF{$a_{l} > b_{l}$}{ $V \gets remove(v_{l},V)$ and $continue;$ }\label{algo:integerize:9}
\ENDIF
\ENDIF
\STATE $l \gets l + 1$
\ENDWHILE 
\STATE $V_{\mathbb{Z}} \gets V$
\end{algorithmic}
\end{algorithm}

\begin{algorithm}
\caption{Remove segment $v_l$ from $V$ - remove$(v_l,V)$.} \label{algo:remove}
\begin{algorithmic}[1]
\FOR{$l' \in \{l \ldots L\}$}
		    \STATE $v_{l'} \gets v_{l'+1} $ \label{algo:remove:1}
\ENDFOR 
\STATE remove $v_L$
\STATE $L \gets L-1$
\end{algorithmic}
\end{algorithm}

\begin{lemma} 
Algorithm \ref{algo:integerize} is correct, i.e.: $V_{\mathbb{Z}} \subset_{\mathbb{Z}} V$
\end{lemma}
\begin{proof}
In step \ref{algo:integerize:1} and step \ref{algo:integerize:2} of Algorithm \ref{algo:integerize}, the domain of the segment is reduced without cutting integers off. In step \ref{algo:integerize:3} obsolete segments are removed.

In step \ref{algo:integerize:4a} all pairs of segments where the domains intersect in one point are identified; in step \ref{algo:integerize:4} and step \ref{algo:integerize:7} the case where one of the segments dominates the other is treated. In case of dominance, the domain of the dominated segment is reduced in step \ref{algo:integerize:5} and step \ref{algo:integerize:8}. In both cases it may happen that the reduced domain vanishes and therefore the corresponding segment will be removed which corresponds to the steps \ref{algo:integerize:6} and \ref{algo:integerize:9}. A detailed procedure for removing segments can be found in Algorithm \ref{algo:remove}. 

Because of the step \ref{algo:integerize:1} and step \ref{algo:integerize:2}, the property \ref{integerize:P1} is satisfied. The other steps may reduce the domain by one but it cannot happen that integers are cut off; therefore \ref{integerize:P2} and \ref{integerize:P3} are satisfied.
\end{proof}

According to the algorithm, each segment of the integerized version $V_{\mathbb{Z}}$ originates in a segment of $V$ with possibly reduced domain. 
In the proof of the main result of this section, the following property will be used:
\begin{lemma} \label{envminintegerized}
If $V' \subset_{\mathbb{Z}} V$ and $W' \subset_{\mathbb{Z}} W$ then:
 $$\min(V',W')_{\mathbb{Z}} \subset_{\mathbb{Z}} \min(V,W)_{\mathbb{Z}}$$
\end{lemma}
\begin{proof}
The value functions $V'$ and $V$ are identical on the corresponding restricted domain, and the same is true for $W'$ and $W$. 
The domain of $\min(V',W')$ is the union of the domains of $V'$ and $W'$ and the domain of $\min(V,W)$ is the union of the domains of $V$ and $W$; therefore the domain of $\min(V',W')$ is contained in the domain of $\min(V,W)$ and the values are identical for points in the domain of $\min(V',W')$. According to that, \ref{integerize:P2} and \ref{integerize:P3} is the consequence. The property \ref{integerize:P1} is true because of integerizing the domains.
\end{proof}
As a consequence, the following corollary can be formulated:
\begin{corollary} \label{envminintegerized2}
If $(V' \boxplus f) \subset_{\mathbb{Z}} (V \boxplus f)  $ and $(W' \boxplus f) \subset_{\mathbb{Z}} (W \boxplus f)$ then:
$$((V' \oplus W') \boxplus f)_{\mathbb{Z}} \subset_{\mathbb{Z}} ((V \oplus W) \boxplus f)$$
\end{corollary}
\begin{proof}
The symbol $\oplus$ denotes the concatenation of segments. Note that $(V \oplus W) \boxplus f = \min(V \boxplus f, W \boxplus f)$, therefore Lemma \ref{envminintegerized} can be used to prove the statement.
\end{proof}

The DP recurrence relation for integerized domains is defined as follows:
\begin{align}
W_{0} = (-f_{1}|_{[0,Q_{max}]})_{\mathbb{Z}} \\
W_{i} = (\overbrace{\min_{j < i} ( W_j - c_{j,i})_{\mathbb{Z}}}^{\tilde{W}_{i} := } \boxplus f_{i})_{\mathbb{Z}}
\end{align}
In the following, the equivalence of $W_{i}$ and $V_i$ for $f_i$ having domains with integer borders will be stated in the following Proposition:
\begin{proposition}
If each segment of $f_i$ has a domain with integer borders, then the following is true:
\begin{enumerate}[label=(A\arabic*)]
\item an optimal integer solution $(y_j)$ (for $j \leq i$) exists for $V_i(q)$ where $q$ is feasible and integer. \label{integerizeproof:P1}
\item $ W_i \subset_{\mathbb{Z}} V_i $ \label{integerizeproof:P2}
\end{enumerate}
\end{proposition}
\begin{proof}
Note that \ref{integerizeproof:P1} means that an optimal integer solution for $V_n(0)$ can be found, which corresponds to an integer optimal solution. 

A proof by induction for \ref{integerizeproof:P1} will be given: for $i = 1$ there is nothing to prove. Suppose that the property holds for $j<i$. By definition: 

\begin{align}
V_{i}(q) = \min_{\substack{ j < i\\ l \leq l_{max}^{V_j} \\ l' \leq {l}^{f_i}_{max}}} \left( \min_{\substack{ q-y \in [a_{l-1}^{V_j},a_l^{V_j}] \\ y \in [a_{l-1}^{f_i},a_l^{f_i}]}}\{ v_l^{V_j}(q-y) + c_{j,i+1} - v_{l'}^{f_{i}}(y) \} \right) \label{fxint:rec}
\end{align}
Therefore for all feasible integers $q$ there exists a $j < i$ and a $y$ such that $V_{i}(q,y) = v_l^{V_j}(q-y) + c_{j,i} - v_{l'}^{f_{i}}(y)$ for $q-y \in [a_{l-1}^{V_j},a_l^{V_j}]$ and $y \in [a_{l'-1}^{f_i},a_{l'}^{f_i}]$. If $y$ is fractional then $\lfloor y \rfloor, \lceil y \rceil \in [a_{l'-1}^{f_i},a_{l'}^{f_i}]$. 

Suppose that $\Gamma^{V_j}$ is the set of feasible states in stage $j$, then $v_l^{V_j}$ is part of the envelope. The set of $\Gamma^{V_j}$ is the union of Minkowski sums of $j$ segments ($j$-tuple) and the envelope of $\Gamma^{V_j}$ is the envelope of these $j$-tuples. The envelope of each $j$-tuple can be represented by the envelope of a set of $j 2^{j-1}$ segments with integer borders (cf. edges of a hypercube) and the union of these segments (when considering all possible tuples) can be used to represent the envelope $\text{env}(\Gamma^{V_j})$.
According to that, we can find a feasible segment $v^{\Gamma^{V_j}}$ with integer borders such that $v^{\Gamma^{V_j}}(q-y) = v_l^{V_j}(q-y)$. Obviously, either $v^{\Gamma^{V_j}}(q-\lceil y \rceil) \leq v_l^{V_j}(q-y)$ or $v^{\Gamma^{V_j}}(q-\lfloor y \rfloor) \leq v_l^{V_j}(q-y)$ and therefore $y$ can be replaced by $\lfloor y \rfloor$ or $\lceil y \rceil$. \\

For proving \ref{integerizeproof:P2} it is sufficient to show that:
\begin{eqnarray}
&\tilde{W}_{i} \subset_{\mathbb{Z}} \tilde{V}_{i} \text{ and }  \label{integerizeproof:P2A}\\
& (\tilde{W}_{i} \boxplus f_{i})  \subset_{\mathbb{Z}} (\tilde{V}_{i} \boxplus f_{i})  \label{integerizeproof:P2B}
\end{eqnarray}
\indent Proof of (\ref{integerizeproof:P2A}) by induction. For $i=1$ there is nothing to prove. Suppose that the above properties hold for $j<i$ and $W_i \subset_{\mathbb{Z}} V_j$. Therefore Lemma \ref{envminintegerized} can be used to show that $\tilde{W}_{i} \subset_{\mathbb{Z}} \tilde{V}_{i}$ because adding a constant does not change the domain of the segments. 

Proof of (\ref{integerizeproof:P2B}) by induction on the number of segments in $\tilde{V}_{i}$. For a single segment it follows that the borders are integer and therefore the superposition with $\tilde{W}_{i}$ is defined by the the border of feasible states for the superposition with $\tilde{V}_{i}$. According to that, each segment in $\tilde{W}_{i} \boxplus f_i$ can be identified by a segment in $\tilde{W}_{i} \boxplus f_i$. This, in combination with \ref{integerizeproof:P1} gives the proof for the first segment.
Suppose the assumption is true for $\tilde{V}_{i}$ with $m$ segments, then it will be shown that the assumption holds for $\tilde{V}_{i} \oplus v_{m+1}^{\tilde{V}}$. Suppose that $\tilde{W}_{i} \subset_{\mathbb{Z}} (\tilde{V}_{i} \oplus v_{m+1}^{\tilde{V}})$, then $\tilde{W}_{i}$ consists of segments that correspond to $\tilde{V}_{i}$ or $v_{m+1}^{\tilde{V}}$; they are denoted $\tilde{W}_{i}^{1 \ldots m}$ and $\tilde{W}_{i}^{m+1}$ respectively.

If $v_{m+1}^{\tilde{V}}$ has integer borders, then $\tilde{W}_{i}^{1 \ldots m} \subset_{\mathbb{Z}}  \tilde{V}_{i} $ and $\tilde{W}_{i}^{m+1} \subset_{\mathbb{Z}} v_{m+1}^{\tilde{V}}$, and it follows by induction that $(\tilde{W}_{i}^{1 \ldots m} \boxplus f_i)  \subset_{\mathbb{Z}}  (\tilde{V}_{i} \boxplus f_i) $ and  $(\tilde{W}_{i}^{m+1} \boxplus f_i)  \subset_{\mathbb{Z}}   (v_{m+1}^{\tilde{V}} \boxplus f_i) $, and therefore
 $(\tilde{W}_{i} \boxplus f_i)  \subset_{\mathbb{Z}}  ((\tilde{V}_{i} \oplus v_{m+1}^{\tilde{V}}) \boxplus f_i) $. 

Alternatively, if $v_{m+1}^{\tilde{V}}$ has a left border $a_{m}^{\tilde{V}}$ that is not integer, then $ \tilde{W}_{i}^{m+1}$ consists of segment that correspond to segments where the domain is bounded to left by $\lceil a_{m}^{\tilde{V}} \rceil$ and $\tilde{W}_{i} = \tilde{W}_{i}^{1 \ldots m}$  correspond to segments where the domain is bounded to the right by $\lfloor a_{m}^{\tilde{V}} \rfloor$. Note that $(\lfloor a_{m}^{\tilde{V}} \rfloor,\lceil a_{m}^{\tilde{V}} \rceil)$ does not contain points of $\tilde{W}_i$, and according to Corollary \ref{envminintegerized2} the proof of \ref{integerizeproof:P2} is complete.

\end{proof}

Figure \ref{fig:segments300} and Figure \ref{fig:segments60} present a comparison of the number of segments in the final value function (last stage) for different $Q_{max}$. According to the example the reduction in the number of segments can be considerably large, especially regarding the outliers.

\begin{figure}[p]
\begin{center}
\includegraphics[scale=0.5]{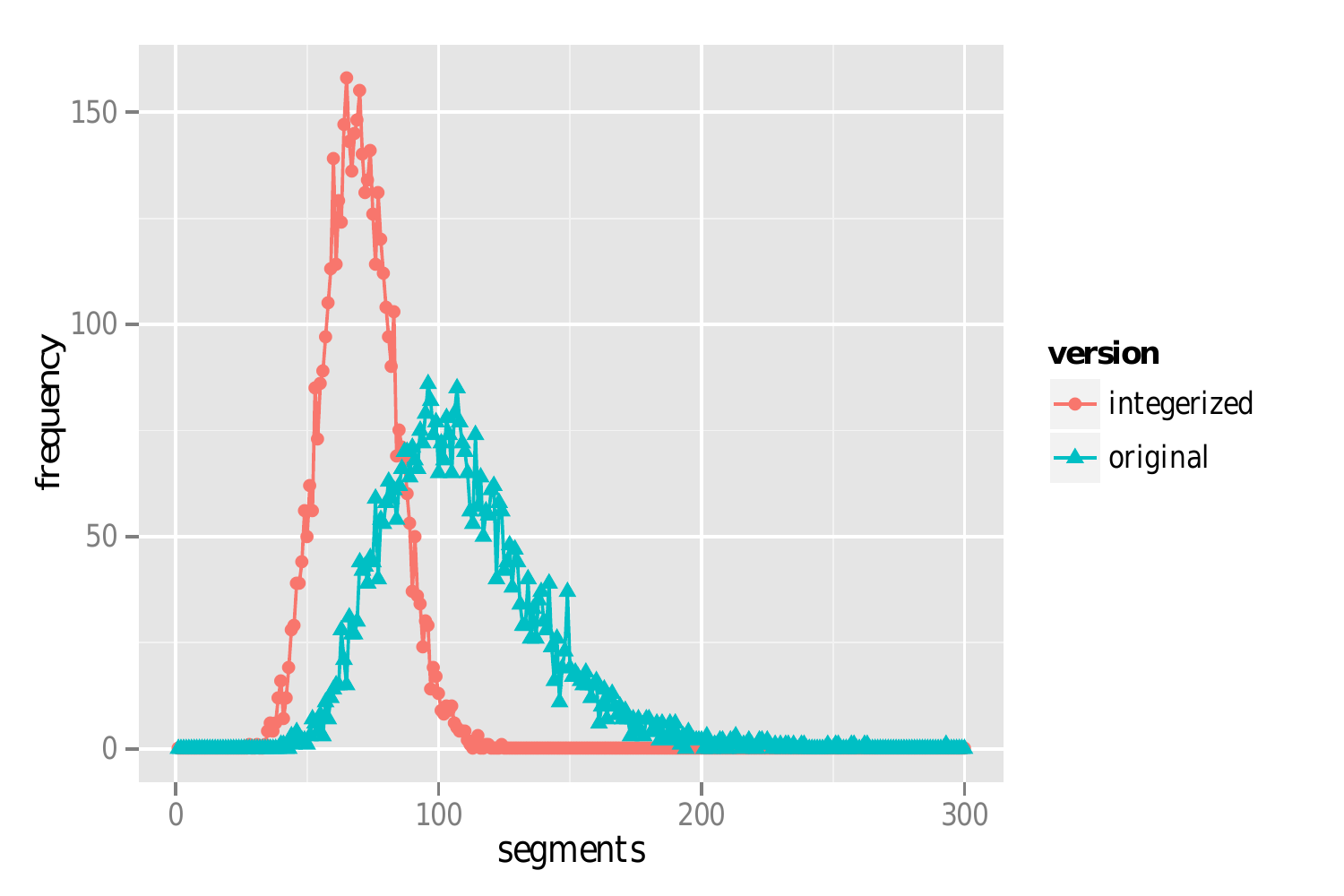}
\end{center}
\caption{frequency for the total number of segments in the final value function for 30 customers and 5000 randomly selected a priori routes and $Q_{max} = 300$}
\label{fig:segments300}
\end{figure}

\begin{figure}[p]
\begin{center}
\includegraphics[scale=0.5]{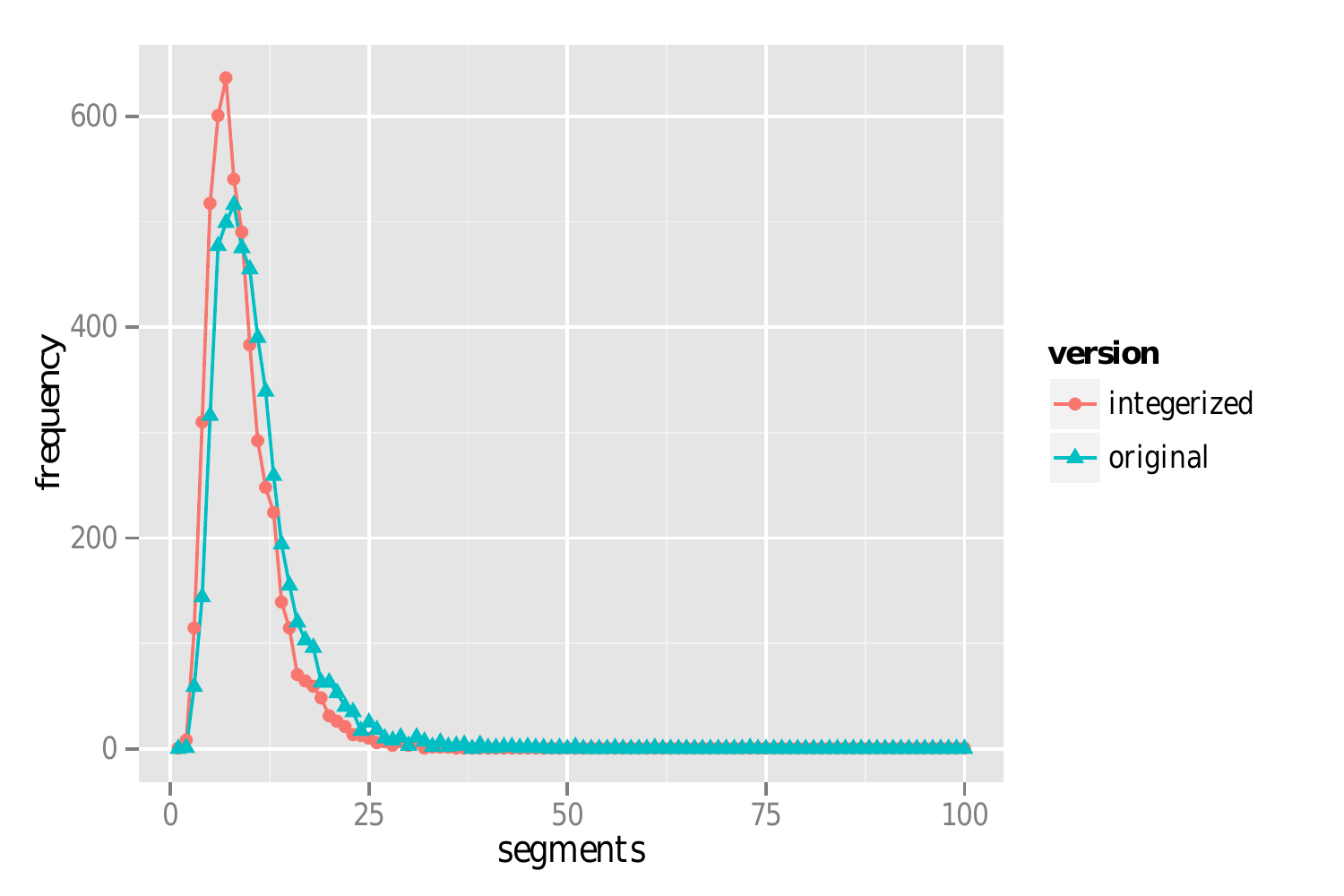}
\end{center}
\caption{Frequency for the total number of segments in the final value function for 30 customers, 5000 randomly selected a priori routes and $Q_{max} = 60$}
\label{fig:segments60}
\end{figure}

\end{appendices}
\end{document}